\documentclass[11pt]{article}

\usepackage{setspace}

\usepackage{amsmath,amsthm, amssymb, latexsym}
\usepackage{amsfonts}
\usepackage{graphicx}
\usepackage{multicol}
\usepackage{caption}
\usepackage{color}
\usepackage{graphicx}
\usepackage{amsmath}
\usepackage{multicol}
\usepackage{float}
\usepackage{color}
\usepackage{subfigure}
\usepackage{url}
\usepackage{tabularx}
\usepackage{xcolor}
\usepackage{grffile}
\usepackage[english]{babel}
\usepackage{hyperref}
\newtheorem{Theorem}{Theorem}[section]

\newtheorem{Lemma}[Theorem]{Lemma}
\newtheorem{Corollary}[Theorem]{Corollary}

\newtheorem{Remark}[Theorem]{Remark}

\newcommand{\E}{\mathbb{E}}

\begin{document}
	
	\title{Stochastic leverage effect in high-frequency data: a Fourier based analysis. \\
	}
	\author{\small{Imma Valentina Curato\footnote{Corresponding author: Ulm University, Institute of Mathematical Finance, Helmholtzstrae 18, 89069 Ulm, Germany. {\sc E-mail:} imma.curato@uni-ulm.de} and Simona Sanfelici\footnote{University of Parma, Department of Economics, Via J. Kennedy, 6, 43125 Parma, Italy. {\sc E-mail:} simona.sanfelici@unipr.it}}}
	
\onehalfspacing
	
	\maketitle
	
	\textwidth=160mm \textheight=225mm \parindent=8mm \frenchspacing
	\vspace{3mm}

\begin{abstract}
The stochastic leverage effect, defined as the standardized covariation between the returns and their related volatility, is analyzed in a stochastic volatility model set-up. A novel estimator of the effect is defined using a pre-estimation of the Fourier coefficients of the return and the volatility processes. The consistency of the estimator is proven. Moreover, its finite sample properties are studied in the presence of microstructure noise effects. The Fourier methodology is applied to S\&P500 futures prices to investigate the magnitude of the stochastic leverage effect detectable at high-frequency.
\end{abstract}
	
	
	
	\section{Introduction}
\label{sec1}

The leverage effect is one of the most striking empirical regularity observed in financial time series. In its classical interpretation given by \cite{B} and \cite{C82}, the effect refers to the negative and constant correlation typically observed between returns and their respective volatilities. \cite{BW00,C87,C92,F87,FW01,G05,N91} and \cite{W01} empirically investigate the presence of constant and negative correlation between returns and volatilities across different financial asset types. They observe that the effect is in general larger for aggregate market index returns than for individual stocks, see discussion in \cite{T96}, and detectable at frequencies lower than or equal to $1$ day. 

In the empirical literature employing high-frequency data, i.e. intra-daily data, there is no consensus in interpreting the leverage effect as above. \cite{AFL13,B06,T96} demonstrate the presence of a constant and negative correlation between returns and volatilities, the analysis in \cite{BR10,FW01,Y05} support the claim of a time-varying effect, and \cite{ASFLWY13,AJ14,CW,KX17,MT20,MW,MZ09} analyze the presence of stochastic correlation between returns and volatilities. We follow this last strand of literature.

In a high-frequency framework, it is more appropriate to define the leverage effect following \cite{AFL13}, namely, as the instantaneous correlation  
\begin{equation}
	\label{rhoi}
	R(t)= \frac{\langle{dp,d\sigma^2} \rangle}{\sqrt{\langle dp,dp \rangle \langle d\sigma^2, d\sigma^2, \rangle}},
\end{equation}
which corresponds to the standardized quadratic covariation between the increments of the logarithmic asset price $p$, i.e. the return process, and the increments of the volatility process $\sigma^2$.

\cite{AFL13} observe that the magnitude of the leverage effect (\ref{rhoi}) detected in the data (using a classical realized covariance estimator) is near zero if we use data in a daily time window, and becomes negative if we use a weekly to a monthly time window. However, the leverage effect should not change its value on different time horizons, as it is an intrinsic feature of the model underlying the data. The authors observe that several sources of bias arise. One is due to the latent (i.e. non-observable) volatility path, and a second one to the presence of microstructure noise. We observe the latter when using data at a frequency higher than $5$ minutes, e.g. \emph{tick-data}. \cite{AFL13} employ different proxies of the volatility path to overcome these problems, namely, local averages of integrated volatility estimators and bias corrections. Unfortunately, the methodology described in \cite{AFL13} only works under the assumption that the volatility is a stationary process and  $R(t)$ is equal to a constant, as in the \cite{HES} model set-up.

To avoid the problems related to the estimation (\ref{rhoi}), \cite[Formula 8.42]{AJ14} examine an alternative measure of the leverage effect. They estimate the standardized quadratic covariation between $p$ and $\sigma^2$ in a time window $[0,T]$, 
\begin{equation}
	\label{rho}
	R_T= \frac{ \langle{p,\sigma^2} \rangle_T}{\sqrt{\langle p,p \rangle_T \langle \sigma^2, \sigma^2 \rangle}_T},
\end{equation}
where the \emph{integrated volatility}, i.e. the quadratic variation of $p$, and the \emph{integrated volatility of volatility}, i.e. the quadratic variation of $\sigma^2$, appear at the denominator, and the \emph{integrated leverage} appears at the numerator. We call $R_T$ the \emph{stochastic leverage effect}. For instance, if we assume that $p$ and $\sigma^2$ follow the Heston model, then $R_T=\rho$; otherwise, $R_T$ is in general a random quantity. 

In this paper, we analyze an estimation methodology for the stochastic leverage effect designed to work in the presence of microstructure noise. To this end, it is crucial to make a clear distinction between an estimator of the integrated leverage (appearing at the numerator of $R_T$) and an estimator of the stochastic leverage effect. The latter is a plug-in estimator that combines estimates of the former, of the integrated volatility and the integrated volatility of volatility. \cite{ASFLWY13} and \cite{MW} discuss the estimation of the integrated leverage in the presence of microstructure noise. In the former, the authors also analyze models with jumps.  \cite{KX17} and \cite{ASFLWY13} employ the plug-in estimator for $R_T$ defined in \cite{AJ14} in different simulation analysis. \cite{ASFLWY13} examine an estimation of $R_T$ in a Heston model set-up and in the absence of microstructure noise effects. Their analysis depends on several bias corrections, especially applied to the volatility of volatility estimation, and the tuning of parameters identifying, for example, the length of the time windows of data used to estimate the latent volatility path. On the other hand, \cite{KX17} use a volatility instrument as the VIX to perform their estimation.

We present a methodology to estimate the stochastic leverage effect, which is based on a continuous-time model. We refer the reader to Remark \ref{tick} for more details on this modelling assumption. Therefore, when comparing our methodology with the state-of-the-art literature, we refer to \cite{MW} for an estimation of the integrated leverage, and to \cite{AJ14} and \cite{ASFLWY13} for discussing theoretical and numerical features of an estimator of $R_T$. The paper of \cite{KX17} is based on a different modelling framework, and we do not consider it further. 

The target of this paper is twofold. First of all, we want to develop an estimation strategy for $R_T$ that constitutes an alternative to the one proposed in \cite{AJ14}, and handles data contaminated by microstructure noise. Secondly, we want to determine a selection strategy for the tuning parameters appearing in our proposed methodology and analyze the performance of the estimator in set-up  other than the Heston model. Notably, this latter point is critical for using \emph{tick-data} because they are not always well described by a Heston model, see Remark \ref{tick}. To the best of our knowledge, the numerical and empirical analysis conducted in the paper is the first examining the presence of the stochastic leverage effect (\ref{rho}) at high frequency in the presence of microstructure noise effects.

Our estimator employs the Fourier methodology introduced in \cite{MM}, see also \cite{MRS} for a complete overview. We call it the \emph{Fourier estimator of the stochastic leverage effect} (in short, FESL). We choose this methodology because it avoids estimating the latent volatility path. This step is mandatory in the estimators appearing in \cite{ASFLWY13}  and is one reason behind several bias corrections applied to their estimations of $R_T$.

To define a Fourier estimator of $R_T$, we then combine three different estimators: the Fourier estimator of the integrated leverage (FEL), of the integrated volatility (FEV) and the integrated volatility of volatility (FEVV) which have been defined in \cite{CS15}, \cite{MM} and \cite{CMS15}, respectively. When estimating the numerator and denominator of (\ref{rho}), we handle the latent volatility by computing $N$ Fourier coefficients of the volatility process. This step requires the preliminary computation of $M$ Fourier coefficients of the returns. We call the parameters $M$ and $N$ \emph{cutting frequency parameters} in the following. In the Fourier set-up described in this paper, $M$ and $N$ play a role similar to the tuning parameters in \cite{ASFLWY13}.

The consistency and the finite sample properties of the FESL are strictly related to a thorough analysis of the consistency and finite sample properties of the FEL, the FEV, and the FEVV used in the estimation. \cite{MM09} analyze the consistency of the FEV, whereas \cite{MS08} study its finite sample properties and a selection strategy for the cutting frequency parameters appearing in the estimation. Regarding the FEVV, \cite{CMS15} analyze its consistency, finite sample properties, and selection of parameters $M$ and $N$. Unfortunately, the theoretical results available for the FEL are not sufficient to directly obtain the consistency of the FESL because the three consistency theorems related to the FEL, the FEV, and the FEVV, available in the literature, hold under different assumptions. Hence, we prove a new consistency result for the FEL as detailed in Section \ref{sec3}. Moreover, the finite sample properties of the FEL, in the presence of microstructure noise,  have not yet been analyzed in the literature. We focus on them in Section \ref{sec5} and conclude that the FEL is asymptotically unbiased, although it has a diverging mean squared error. 

We propose a variance corrected estimator of the FEL in the presence of microstructure noise. Moreover, in an extensive simulation study, we analyze selection strategies for the parameters $M$ and $N$ appearing in the FEL and its variance corrected version. We use Monte-Carlo data sets drawn from \cite{HES}, and the generalized Heston model presented in \cite{VV} and study how the selection of the parameters $M$ and $N$ impacts the mean squared error and the sample variance of the estimation. Our findings suggest that the parameters obtained by minimizing the mean squared error of the FEL are equivalent to those obtained minimizing its sample variance. Moreover, we note that using the variance corrected estimator reduces the sample variance of the final estimation by a half. Having a selection strategy for the parameters $M$ and $N$, we show a comparison between the performance of the FEL and the realized covariance-based estimator of the integrated leverage presented by \cite{MW}. To conclude, we also perform a sensitivity analysis on the FEL on data sets generated from the generalized Heston model defined in \cite{VV}: note that $R_T$ is a random quantity in this set-up.

The paper has the following structure. In Section \ref{sec3}, we introduce the model set-up, the definition of the FEL, and the FESL together with their asymptotic properties in the absence of microstructure noise. In Section \ref{sec5}, we analyze the finite sample properties of the FEL in the presence of microstructure noise and the selection strategy for the cutting frequency parameters. In Section \ref{sec6}, we discuss how to implement the FESL in the presence of microstructure noise. Section \ref{sec7} applies our results to S\&P500 futures prices. Section \ref{sec8} concludes. The Appendix contains the proofs of all statements presented in the paper.

\section{Estimation of the stochastic leverage effect in the absence of microstructure}
\label{sec3}

We assume throughout that the logarithmic asset price and the volatility process are a solution to the system of equations
\begin{equation}
	\label{mod}
	\left\{ \begin{array}{ll}
		dp(t)&= a(t) \, dt + \sigma(t) \, dW(t)  \\
		d\sigma^2(t)&= b(t) \, dt + \gamma(t) \, dZ(t),
	\end{array} \right. 
\end{equation}
where $W(t), \, t \geq 0$ and $Z(t), \, t \geq 0$ are two correlated standard Brownian motions. Their correlation process is $\rho(t)$ with values in $[-1,1]$. We consider $p(t)$ the underlying \emph{efficient} logarithmic price process.

\begin{Remark}
	\label{tick}
	The choice of a continuous-time modelling set-up for $p(t)$ is motivated by the empirical work of \cite{COP14} where the authors analyze the presence of jumps in \emph{tick data}. They observe that our ability to distinguish true discrete jumps from continuous diffusive variation diminishes as we increase the sampling frequency.
	For example, a short-lived burst of volatility is likely to be identified as a jump when working with data sampled at a frequency lower than $5$ minutes but it is compatible with a continuous path when working with tick data. Thus, we do not consider jumps in our model. We add instead different randomness sources in the dynamics of the logarithmic asset price and its respective volatility that aim to describe the variability observed in \emph{tick data}. An exemplary model in this set-up is the generalized Heston model defined by \cite{VV}.
\end{Remark}

We perform our analysis in a time window $[0,T]$ for $T>0$, and such that the processes appearing in model (\ref{mod}) satisfy the following assumption:
\begin{itemize}
	\item {\bf (H1)} $a(t)$, $b(t)$, $\sigma(t)$, $\gamma(t)$ and $\rho(t)$ are $\mathbb{R}$-valued processes, almost surely continuous on $[0,T]$ such that
	
	\[
	\mathbb{E}\Big[ \sup_{t\in[0,T]}|a(t)|^2\Big]<\infty,  \, \, \, \,
	\,\mathbb{E}\Big[ \sup_{t\in[0,T]}|b(t)|^2\Big]<\infty,
	\]
	\[
	\mathbb{E}\Big[ \sup_{t\in[0,T]}|\sigma(t)|^4\Big]<\infty, \, \,\, \,\,
	\mathbb{E}\Big[ \sup_{t\in[0,T]}|\gamma(t)|^4\Big]<\infty,
	\]
	\[
	\mathbb{E}\Big[ \sup_{t\in[0,T]}|\rho(t)|^2\Big]<\infty.
	\]
\end{itemize}

We start by developing an estimation strategy for (\ref{rho}) in the absence of microstructure noise and studying its consistency. We aim to define a plug-in estimator in Section \ref{rt} which employs the FEL, the FEV, and the FEVV, respectively defined in  \cite{CS15}, \cite{MM}, and \cite{CMS15}.
As the first step, we analyze under which set of assumptions the estimators above are all consistent. Let 
$ {\cal S}_n := \{0=t_{0}\leq t_{1} \leq \ldots \leq t_{n}=T \},$ be the set of observation times,
and define $\tau(n)=\max_{i=0,\ldots,n-1} |t_{i+1} -t_{i}|$. The FEV and FEVV are both consistent under the assumptions that $\frac{N^4}{M} \to 0$, and $M \tau(n) \to 0$ as $n,M, N \to \infty$ and $\tau(n) \to 0$, see \cite{CMS15} and \cite{MM09}. Unfortunately, the consistency of the FEL has not been proved in \cite{CS15} under an assumption of type $M \tau(n) \to 0$. Hence, to define a consistent Fourier estimator of (\ref{rho}), we need to prove that the FEL is consistent under a new set of assumptions, see Remark \ref{new_ass}. To start with, we briefly remind the definition of the FEL.

\subsection{Fourier estimator of the integrated leverage}
Let $(p(t),\sigma^2(t))$ be a solution to (\ref{mod}). We define the leverage process $\eta(t)$ as
\begin{equation}
	\label{def}
	\langle dp(t),d\sigma^2(t) \rangle=\sigma(t)\gamma(t)\rho(t) dt= \eta(t) dt.
\end{equation}
We are interested in estimating the integrated covariation between the logarithmic price and the volatility process, which appears at the numerator of (\ref{rho}), by determining an estimator for
\begin{equation}
	\label{I}
	\eta=\int_{0}^{T} \eta(t) dt.
\end{equation}

We follow a methodology based on the use of the Fourier coefficients of the process $\eta(t)$.

Following \cite{MM}, we define the Fourier coefficients of the returns and of the increments of the volatility process as
\begin{equation}
	\label{cdp}
	c(l;dp)= \frac{1}{T} \int_{0}^{T} \mathrm{e}^{-\mathrm{i}\frac{2\pi}{T}l t} dp(t),
\end{equation}
and
\begin{equation}
	\label{cdv}
	c(l;d\sigma^2)= \frac{1}{T} \int_{0}^{T} \mathrm{e}^{-\mathrm{i}\frac{2\pi}{T}l t} d\sigma^2(t),
\end{equation}
for each $l \in \mathbb{Z}$.
Note that for all $l\neq0$ and by using the integration by parts formula, we can rewrite (\ref{cdv}) as
\begin{equation}
	\label{mod3}
	c(l;d\sigma^2)= \mathrm{i}l \frac{2\pi}{T} c(l;\sigma^2) + \frac{1}{T}(\sigma^2(T)-\sigma^2(0)),
\end{equation}
where
\[
c(l;\sigma^2)=\frac{1}{T} \int_{0}^{T} \mathrm{e}^{-\mathrm{i}\frac{2\pi}{T}l t} \sigma^2(t) dt.
\]

Given two functions $\Phi$ and $\Psi$ on the integers $\mathbb{Z}$,
we say that their Bohr convolution product exists if the following limit exists
for all integers $h$
\[
(\Phi \ast \Psi)(h):= \lim_{N \to \infty} \frac{1}{2N+1} \sum_{|l|\leq N}
\Phi(l) \Psi(h-l).
\]

Under Assumption (H1) and for a fixed $h$, we define $\Phi(l) := c(l;d\sigma^2)$ and $\Psi(h-l) :=c(h-l,dp)$, then the limit in probability of the Bohr convolution product exists and converges to the $h$-th Fourier coefficient of the leverage process.
This result immediately follows from \cite[Theorem 2.1]{MM09}.
The $h$-th Fourier coefficient of $\eta(t)$ is then defined as
\begin{equation}
	\label{limit}
	c(h;\eta) = \lim_{N \to \infty} \frac{T}{2N+1} \sum_{|l|\leq N} c(l;d\sigma^2)
	c(h-l;dp)= \frac{1}{T} \int_{0}^{T} \mathrm{e}^{-\mathrm{i}\frac{2\pi}{T}h t} \eta(t) dt.
\end{equation}

The definition of the Fourier coefficients of $\eta(t)$ has the obvious drawback to being feasible only when continuous observations of the logarithmic price and the volatility paths are available.
By using the methodology described in \cite{CS15,C18}, it is possible to give an estimator of the Fourier coefficients of $\eta(t)$ when discrete and non-equidistant observations of $p(t)$ are available on the time grid ${\cal{S}}_n$ and the volatility is latent. Hereafter, we indicate the discrete observed returns by $\delta_{i}=p(t_{i+1})-p(t_{i})$ for all $i=0,...,n-1$.

An estimator of the $h$-th Fourier coefficient of the leverage process can be defined as
\begin{equation}
	\label{f_eta}
	c_{n,M,N}(h;\eta)= \frac{T}{2N+1} \sum_{|l|\leq N} \mathrm{i}l \frac{2\pi}{T} c_{n,M}(l;\sigma^2)c_{n}(h-l;dp),
\end{equation}
for any integer $h$ such that $|h| \leq N$, where $c_n(s;dp)$ are the discrete Fourier coefficients of the returns
\begin{equation}
	\label{d1}
	c_n(s;dp)= \frac{1}{T} \sum_{i=0}^{n-1} \mathrm{e}^{-\mathrm{i}s\frac{2\pi}{T}t_{i}}
	\delta_{i}(p)
\end{equation}
for $|s|\leq  N+M$, and $c_{n,M}(h;\sigma^2)$ are the Fourier coefficients of the volatility process introduced in \cite{MM} for $|l| \leq  N$
\begin{equation}
	\label{d3_alt}
	c_{n,M}(l;\sigma^2)=\frac{T}{2M+1} \sum_{|s|\leq M}  c_{n}(s;dp)c_{n}(l-s;dp).
\end{equation}

The estimators are written as functions of $n$, $M$, and $N$ which stand for the number of observations available, the number of the discrete Fourier coefficients of the returns, and of the Fourier coefficients of the volatility process, respectively.

Finally, the FEL is obtained from Definition (\ref{f_eta}) for $h=0$ 
\[
\hat{\eta}_{n,M,N}= T c_{n,M,N}(0;\eta)= \frac{T^2}{2N+1} \sum_{|l|\leq N} \mathrm{i}l \frac{2\pi}{T} c_{n,M}(l;\sigma^2)c_{n}(-l;dp).
\]
We can also give a more explicit form of the FEL. By employing the normalized Dirichlet kernel 
\begin{equation}
	\label{dir}
	D_N(t)= \frac{1}{2N+1} \sum_{|l| \leq N} \mathrm{e}^{\mathrm{i}\frac{2\pi}{T}lt},
\end{equation}
and its first derivative 
\begin{equation}
	\label{der}
	D_N^{\prime}(t)= \frac{1}{2N+1} \sum_{|l| \leq N} il \frac{2\pi}{T} \mathrm{e}^{\mathrm{i}\frac{2\pi}{T}lt},
\end{equation}
we obtain
\begin{equation}
	\label{decomp}
	\hat{\eta}_{n,M,N}= \sum_{i=0}^{n-1} \sum_{j=0}^{n-1} \sum_{k=0}^{n-1} D_M(t_{i}-t_{j}) D_N^{\prime}(t_{k}-t_{j})
	\delta_{i} \delta_{j} \delta_{k}.
\end{equation}

We now focus on showing the consistency of the estimator (\ref{decomp}). 

\begin{Remark}
	\label{new_ass}
	In \cite{CS15}, it is shown that the FEL (\ref{decomp}), in the absence of microstructure noise, is consistent if $N^2/M \to 0$ and $M\tau(n)\to a$ with $a >0$ as $n,M,N \to \infty$ and $\tau(n) \to 0$. We now prove consistency under the assumptions $N^2/M \to 0$ and $MN\tau(n)\to 0$ as $n,M,N \to \infty$ and $\tau(n) \to 0$. Note, that the assumption $MN\tau(n)\to 0$ implies that $M \tau(n) \to 0$ which is the assumption under which the consistency of the FEV and the FEVV holds. Changing the asymptotic rate between the parameters $M$, $N$, and $\tau(n)$ implies a complete different consistency's proof of the FEL with respect to the one given by \cite{CS15}. In particular, we do not employ Malliavin calculus as in the proof of \cite[Theorem 3.1]{CS15}.
\end{Remark}

\begin{Theorem}
	\label{T0}
	We assume that Assumption (H1) and 
	\begin{equation}
		\label{Star2}
		\frac{N^2}{M} \to 0 \,\,\,\,\,\textrm{and} \,\,\,\,\, MN\tau(n)\to 0
	\end{equation}
	hold true as $n,M,N \to \infty$ and $\tau(n) \to 0$. 
	Then
	\begin{equation}
		\label{Equi}
		\hat{\eta}_{n,M,N} \xrightarrow{\mathbb{P}} \eta.
	\end{equation}
	
\end{Theorem}

\subsection{Fourier estimator of the stochastic leverage effect}
\label{rt}

We can now define the Fourier estimator of the stochastic leverage effect (FESL) by combining the FEL, the FEV and the FEVV. The FESL takes the form
\[
\hat{R}_T=\frac{\hat{\eta}_{n,M,M}}{ \sqrt{\hat{\sigma}_{n,M}^2 \hat{\gamma}_{n,M,N}^2}},
\]
where 
\begin{equation}
	\label{vol}
	\hat{\sigma}^2_{n,M}= \frac{T^2}{2M+1} \sum_{|s|\leq M}  c_{n}(s;dp)c_{n}(-s;dp),
\end{equation}
and 
\begin{equation}
	\label{volofvol}
	\hat{\gamma}_{n,M,N}^2= \lim_{N \to \infty} \frac{T^2}{2N+1} \sum_{|l|\leq N} \Big(1-\frac{|l|}{N}\Big) l^2 \frac{4\pi^2}{T^2} c_{n,M}(l;\sigma^2) c_{n,M}(-l;\sigma^2).
\end{equation}

Finally, we obtain the consistency of the estimator $R_T$ by using the continuous mapping theorem and the results proved in Theorem \ref{T0}, Theorem 3.2 in \cite{CMS15}, and Theorem 3.4 in \cite{MM09}.

\begin{Corollary}
	\label{Tplug}
	We assume that Assumption (H1) and 
	$$ \frac{N^4}{M} \to 0 \,\,\,\,\,\textrm{and} \,\,\,\,\, MN\tau(n)\to 0$$
	hold true as $n,M,N \to \infty$ and $\tau(n) \to 0$. Then
	$$\hat{R}_{T} \xrightarrow{\mathbb{P}} R_T.$$
\end{Corollary}

\begin{Remark}
	To obtain a central limit theorem for an estimator of $R_T$ is necessary to have central limit theorem results for each estimator appearing in its definition. These results have to hold under the same set of assumptions. A central limit theorem for an estimator of $R_T$ follows from applying the delta method, see \cite{AJ14}. All this said, in the realized covariance-based literature, explicit calculations leading to a central limit theorem of an estimator for $R_T$ are not present. Proving a central limit theorem for the estimator $\hat{R}_T$ is also outside the scope of our paper. However, we plan to analyze the latter in future research projects employing the FESL.
\end{Remark}

\section{Finite sample properties of the Fourier estimator of the integrated leverage}
\label{sec5}

The finite sample properties of the FEV and the FEVV have been analyzed in the presence of microstructure noise in \cite{MS08} and \cite{CMS15}, respectively. We study in this section the finite sample properties of the FEL. 

The results contained in this section hold under the following assumption.

\begin{itemize}
	\item {\bf (H2)} $a(t)$, $b(t)$, $\sigma(t)$, $\gamma(t)$ and $\rho(t)$ are $\mathbb{R}$-valued processes, almost surely continuous on $[0,T]$ such that
	
	\[
	\mathbb{E}\Big[ \sup_{t\in[0,T]}|a(t)|^8\Big]<\infty,  \, \, \, \,
	\,\mathbb{E}\Big[ \sup_{t\in[0,T]}|b(t)|^8\Big]<\infty,
	\]
	\[
	\mathbb{E}\Big[ \sup_{t\in[0,T]}|\sigma(t)|^8\Big]<\infty, \, \,\, \,\,
	\mathbb{E}\Big[ \sup_{t\in[0,T]}|\gamma(t)|^8\Big]<\infty,
	\]
	\[
	\mathbb{E}\Big[ \sup_{t\in[0,T]}|\rho(t)|^8\Big]<\infty.
	\]
\end{itemize}

Moreover, we add microstructure noise to the underlying efficient logarithmic price $p(t)$ defined in (\ref{mod}) by assuming that the logarithm of the observed price is
\begin{equation}
	\label{noiseL}
	\widetilde p(t_{i})= p(t_{i})+ \zeta(t_{i}),\,\,\,\,\,\textrm{for $i=0, \ldots, n$},
\end{equation}
where $\zeta(t)$ is the microstructure noise.
We also assume the following

\begin{itemize}
	\item {\bf (H3)} The random shocks $(\zeta(t_i))_{0\leq i\leq n}$ are independent and identically distributed with bounded sixth moment. Moreover, the random shocks are independent of $p(t)$.
\end{itemize}
\vspace{3mm}

We define $\epsilon_i=\zeta(t_{i+1})-\zeta(t_i)$ and $\widetilde{\delta}_i=\widetilde p(t_{i+1})-\widetilde p(t_i)$.
Then, the FEL (\ref{decomp}) in the presence of microstructure noise becomes

\begin{equation}
	\label{decomp_noise}
	\widetilde{\eta}_{n,M,N}=  \sum_{i=0}^{n-1} \sum_{j=0}^{n-1} \sum_{k=0}^{n-1} D_M(t_i-t_{j}) D_N^{\prime}(t_{k}-t_{j})
	\widetilde{\delta}_{i} \widetilde{\delta}_{j} \widetilde{\delta}_{k}.
\end{equation}
We can disentangle (\ref{decomp_noise}) as
\begin{align}
	&\sum_{i\neq j \neq k}   D_M(t_i-t_{j}) D_N^{\prime}(t_{k}-t_{j}) \widetilde{\delta}_{i} \widetilde{\delta}_{j} \widetilde{\delta}_{k} \label{control1} \\
	&+\sum_{i,j:i\neq j} D_M(t_i-t_{j}) D_N^{\prime}(t_{i}-t_{j}) \widetilde{\delta}_{i}^2 \widetilde{\delta}_{j} + \sum_{i,j} D_N^{\prime}(t_{i}-t_{j})  \widetilde{\delta}_{i} \widetilde{\delta}_{j}^2 \nonumber\\
	= &\sum_{i\neq j \neq k}   D_M(t_i-t_{j}) D_N^{\prime}(t_{k}-t_{j}) \delta_{i} \delta_{j}  \delta_{k} \label{dec1}\\
	+ & \sum_{i,j:i\neq j} D_M(t_i-t_{j}) D_N^{\prime}(t_{i}-t_{j}) \delta_{i}^2 \delta_{j} + \sum_{i,j} D_N^{\prime}(t_{i}-t_{j})  \delta_{i} \delta_{j}^2  \label{dec2}\\
	+& \eta^{\epsilon}_{n,M,N}, \nonumber
\end{align}
where the sum of the components (\ref{dec1}) and (\ref{dec2}) corresponds to the FEL in the absence of microstructure noise and all the noise components are contained in $\eta^{\epsilon}_{n,M,N}$. The explicit expression of the latter can be found in the Appendix, see  (\ref{eta_noise}).

In the modeling set-up (\ref{mod}), the integrated leverage can be positive or negative. Therefore, we analyze the bias of the estimator in absolute value. The definition of the FEL does not require the use of equidistant data. However, for simplicity of computation, we assume equidistant observations in the time window $[0,T]$.
\begin{Theorem}
	\label{bias}
	We assume that Assumptions (H2), (H3)  and
	\begin{equation}
		\label{hyp1}
		\frac{N^2}{M}\to 0 \,\,\,\textrm{and}\,\,\, \frac{MN}{n} \to 0
	\end{equation}
	hold true as $N,M,n \to \infty$. Then the estimator $\widetilde{\eta}_{n,M,N}$ is asymptotically unbiased. More precisely,
	
	\[
	\Big |  \E[\widetilde{\eta}_{n,M,N}-\eta] \Big | \leq \Big | \E[\eta_{n,M,N}-\int_0^{T} \eta(t)\, dt]  \Big |+ \Big |\E[\eta^{\epsilon}_{n,M,N}] \Big |
	\]
	\[
	\leq  \Gamma(n,M,N)+ \Lambda(n,N)+ \Psi(N)+ \Big | 2(n-1)\Big(D_M\Big(\frac{T}{n}\Big)-1\Big)D_N^{\prime}\Big(\frac{T}{n}\Big) \E[\zeta^3]\Big |,
	\]
	where
	\[
	\Gamma(n,M,N) \leq \frac{N(M+N)}{n} \, 8\pi^2 T^{\frac{1}{2}}\, \,\E\Big[\sup_{[0,T]} \sigma^2(t)\Big]^{\frac{3}{2}} + \frac{N}{\sqrt{2M+1}} \, 2 \pi\, T^{\frac{1}{2}} \, \E\Big[\sup_{[0,T]} \sigma^2(t)\Big]^{\frac{3}{2}},
	\]
	\[
	\Lambda(n,N) \leq \frac{N}{\sqrt{n}} \,  4\pi T^{\frac{1}{2}} \,  \E\Big[\sup_{[0,T]}\sigma^2(t)\Big]^{\frac{3}{2}} + \frac{N^2}{n} \, 4\pi^2(1+T^{\frac{1}{2}})\, \E\Big[\sup_{[0,T]} \sigma^4(t)\Big]^{\frac{1}{2}} \E\Big[\sup_{[0,T]}\sigma^2(t)\Big]^{\frac{1}{2}},
	\]
	and,
	\[
	\Psi(N)\leq  \frac{1}{\sqrt{2N+1}} \, T \,\E\Big[\sup_{[0,T]} \eta^2(t)\Big]^{\frac{1}{2}}.
	\]
	
\end{Theorem}

The mean squared error of the estimator (\ref{decomp_noise}) has the following bias-variance decomposition
\begin{equation}
	\label{mse_decomp}
	\E[(\widetilde{\eta}_{n,M,N}-\eta)^2]= Var(\widetilde{\eta}_{n,M,N})+ \E[\widetilde{\eta}_{n,M,N}-\eta ]^2 +Var(\eta)-2 Cov(\widetilde{\eta}_{n,M,N},\eta).
\end{equation}
This decomposition differs from the classical bias-variance decomposition which can be found in the parametric statistics literature because the quantity we aim to estimate, i.e. the integrated leverage $\eta$, is a random variable and not a constant parameter.

\begin{Theorem}
	\label{MSE}
	We assume that Assumptions (H2), (H3)  and
	\begin{equation}
		\label{hyp2}
		\frac{N^2}{M} \to 0 \,\,\,\textrm{and}\,\,\,\frac{MN}{n} \to 0
	\end{equation}
	hold true as $N,M,n \to \infty$. Then
	\[
	\E[(\widetilde{\eta}_{n,M,N}-\eta)^2] \to \infty.
	\]
\end{Theorem}
The theorem above highlights a divergent element in (\ref{mse_decomp}) that we try to identify using a numerical analysis in the next section. A diverging mean squared error in the presence of microstructure noise is a phenomenon already observed for integrated estimators in a high-frequency setting. For example, the variance of the realized volatility estimator diverges in the presence of microstructure noise effects as discussed by \cite{BR05}. In this framework, the presence of microstructure noise is usually handled by using pre-averaging, see \cite{J09}. This methodology is used in \cite{ASFLWY13} and \cite{MW} to define estimators of the integrated leverage robust to microstructure noise. However, the Fourier approach automatically filters the noise components on its own, see discussion in \cite[Chapter 5]{MRS}. Therefore, correcting our estimation by using a pre-averaging approach is not an option. 

We define instead a variance corrected version of (\ref{decomp_noise}). We call the term (\ref{control1}) by $\Upsilon_{n,M,N}$. This term contains all the cross products of the noisy returns $\tilde \delta_i \tilde \delta_j \tilde \delta_k$ with $i\neq j\neq k$ and is correlated to $\widetilde{\eta}_{n,M,N}$. Moreover, it has expected value equal to zero as shown in the Corollary below.
\begin{Corollary}
	\label{zero_cor}
	We assume that the assumptions of Theorem \ref{bias} hold. Then, the addend (\ref{control1})  has expected value equal to zero.
\end{Corollary}

\begin{proof}
	The term (\ref{control1}) can be decomposed as
	\begin{align*}
		&\sum_{i\neq j \neq k}   D_M(t_i-t_{j}) D_N^{\prime}(t_{k}-t_{j}) \delta_{i} \delta_{j} \delta_{k}\\
		&+\sum_{i, j, k:  i\neq j \neq k} D_M(t_{i}-t_{j}) D_N^{\prime}(t_{k}-t_{j})
		(\delta_i\delta_j\epsilon_k+ \delta_j \delta_k\epsilon_i+ \delta_k \delta_i\epsilon_j +\delta_i\epsilon_j\epsilon_k +\delta_j\epsilon_i\epsilon_k\\
		&+\delta_k\epsilon_i\epsilon_j+\epsilon_i\epsilon_j\epsilon_k).
	\end{align*}
	
	The thesis follows straightforwardly from the results contained in the proof of Theorem \ref{bias}.
\end{proof}

We then define the estimator
\begin{equation}
	\label{ControlVar3}
	\eta^*_{n,M,N}=\widetilde{\eta}_{n,M,N}-b \, \Upsilon_{n,M,N}
\end{equation}
which is asymptotically unbiased because of Theorem \ref{bias} and Corollary \ref{zero_cor} and has
\begin{equation}
	\label{cv}
	Var(\eta^*_{n,M,N})=Var(\widetilde{\eta}_{n,M,N}) -2 \, b\, Cov(\widetilde{\eta}_{n,M,N},\Upsilon_{n,M,N}) +b^2 \, Var(\Upsilon_{n,M,N}).
\end{equation}
Hence, the estimator $\eta^*_{n,M,N}$ has smaller variance than the estimator $\widetilde{\eta}_{n,M,N}$ provided that
\[
b^2 \, Var(\Upsilon_{n,M,N})<2\, b\, Cov(\widetilde{\eta}_{n,M,N},\Upsilon_{n,M,N}).
\]
The optimal coefficient $b^*$ minimizing the variance of the estimator $\eta^*_{n,M,N}$ is given by
\begin{equation}
	b^*_{M,N}=\frac{Cov(\widetilde{\eta}_{n,M,N},\Upsilon_{n,M,N})}{Var(\Upsilon_{n,M,N})}.
	\label{optimal_b}\end{equation}
Plugging this value in (\ref{cv}) and simplifying, we find
\[
\frac{Var(\eta^*_{n,M,N})}{Var(\widetilde{\eta}_{n,M,N})}=(1-Corr(\widetilde{\eta}_{n,M,N},\Upsilon_{n,M,N})^2),
\]
which gives us the variance reduction ratio obtained by using the estimator (\ref{ControlVar3}).
\begin{Remark}
	A similar variance reduction appears, for instance, in the classical control variate method to reduce the variance of the sample mean estimator, see (\cite[Section 4.1]{G04}).
\end{Remark}

In the next section, we determine selection strategies for the cutting frequency parameters $M$ and $N$ of the FEL and its variance corrected version.

\subsection{Selection strategy for the cutting frequency parameters: a numerical study}
\label{sec5.1}

We assume two different models for the underlying efficient price process, i.e. the classical model proposed by \cite{HES} and the generalized Heston model proposed by \cite{VV}. The microstructure noise satisfies Assumption (H3).

We simulate second-by-second return and variance paths over a daily trading
period of $T=6$ hours,
for a total of $100$ trading days and $n=21600$ observations per day.\\
The first data generating process is
\begin{equation}
	H: \,\,\,\, \left\{ \begin{array}{ll}
		dp(t)&=\sigma(t) dW_1(t)\\
		d\sigma^2(t)&= \alpha(\beta- \sigma^2(t))dt+ \nu \sigma(t) dW_2(t),
	\end{array} \right.
	\label{HestonModel}\end{equation}
where $W_1$ and $W_2$ are correlated Brownian motions.
The parameter values used in the simulations are
$\alpha=0.01, \beta=0.2,\nu=0.05$ and the correlation parameter is set to
$\rho=-0.2$.\\
The second data generating process is
\begin{equation}
	GH:\,\,\,\,\left\{ \begin{array}{ll}
		dp(t)&=\sigma(t) dX(t)\\
		dX(t)&=\rho(t)dW_1(t)+\sqrt{1-\rho^2(t)}dW_2(t)\\
		d\sigma^2(t)&= \alpha(\beta- \sigma^2(t))dt+ \nu \sigma(t) dW_1(t),
	\end{array} \right.
	\label{GHmodel}
\end{equation}
and the infinitesimal variation of $\rho(t)$ is given by
\[
d\rho(t)=((2\xi-\eta)-\eta \rho(t))dt+\theta \sqrt{(1+\rho(t))(1-\rho(t))} dW_0,
\]
where $\eta, \xi$ and $\theta$ are positive constants and $W_0$ is a Brownian
motion.
The processes $W_0(t),W_1(t)$ and $W_2(t)$ are assumed to be independent.
The parameter values used in the simulation are $\alpha=0.01,
\beta=0.2,\nu=0.05$ and $\xi=0.02,\eta=0.5, \theta=0.5$, where the last
three parameters are chosen in the range prescribed
by \cite{VV} such that $\rho(t) \in [-1,1]$.
We set the initial values as $\sigma^2(0)=\beta$, $p(0)=\log(100)$ and
$\rho(0)=-0.04$.
The noise-to-signal ratio $std(\zeta)/std(r)$ is equal to 0.8, where $r$ is the $1$-second returns.

When processing simulated data, the natural approach in optimizing estimators depending on tuning parameters is to choose those values that minimize the finite sample mean squared error (MSE). Therefore, one possible choice is to select the cutting frequency parameters $M$ and $N$ by following this methodology.  We analyze two types of MSE-based optimal strategies. The first directly minimizes the MSE of the FEL $\widetilde{\eta}_{n,M,N}$, whereas the second one is described below.

Operatively, the variance corrected estimator (\ref{ControlVar3}) can be implemented by the following procedure:
\begin{enumerate}
	\item[\textbf{Step 1:}] Given a sample of $n$ observed returns and for all $M \in \{ range \}$ and $N \in \{ range\}$, let $\widetilde{\eta}_{n,M,N}^1, \widetilde{\eta}_{n,M,N}^2, \ldots, \widetilde{\eta}_{n,M,N}^d$ be $d$  replications of the Fourier estimate of the integrated leverage in a Monte Carlo experiment. Along with $\widetilde{\eta}_{n,M,N}^i$, on each replication we also calculate $\Upsilon_{n,M,N}^i$;
	\item[\textbf{Step 2:}] let $M^*,N^*:=\mbox{argmin VAR}(\Upsilon_{n,M,N})$ and let $\Upsilon^*:=\Upsilon_{n,M^*,N^*}$; 
	\item[\textbf{Step 3:}] plug the selected correction $\Upsilon^*$ into equation (\ref{ControlVar3})
	$$\eta^*_{n,M,N}=\widetilde{\eta}_{n,M,N}-b^*_{M,N} \, \Upsilon^*,$$
	where $$b^*_{M,N}=\frac{\mbox{COV}(\widetilde{\eta}_{n,M,N},\Upsilon^*)}{\mbox{VAR}(\Upsilon^*)}$$
	and $\mbox{COV}$, $\mbox{VAR}$ denote the sample covariance and the sample variance, respectively.
	For each $M$ and $N$, compute $d$ replications $\eta^{i*}_{n,M,N}$ ($i=1,\ldots,d$) of the estimator;
	\item[\textbf{Step 4:}] choose the cutting frequency parameters $\hat M$ and $\hat N$ which minimize the finite sample MSE of the corrected estimates $\eta^{i*}_{n,M,N}$ for $i=1,\ldots,d$.
\end{enumerate}

The magnitude of the variance correction given by the estimator (\ref{ControlVar3}) is tuned by formula (\ref{optimal_b}), where $Var(\Upsilon_{n,M,N})$ appears at the denominator. We first set the parameter (\ref{optimal_b}) to minimize the denominator and enhance the effectiveness of the correction.
Afterwards, in Step 4, we choose the optimal MSE-based cutting frequency parameters $\hat M$ and $\hat N$. This procedure provides better empirical results than optimizing the parameters $M$ and $N$ in (\ref{ControlVar3}) simultaneously.

Table \ref{control_eff_noise} shows the MSE reduction obtained by using the estimator
(\ref{ControlVar3}) versus the FEL (\ref{decomp_noise}). The parameter values $\hat M$ and $\hat N$ are selected following the MSE-based optimal strategies described above. 
\begin{table}[h!]
	\centering
	\begin{tabular}{lcccccccc}
		\hline
		& \multicolumn{4}{c}{$H- model$} &
		\multicolumn{4}{c}{$GH- model$} \\
		\hline
		$\overline{\eta}$ & \multicolumn{4}{c}{-1.013673e-04 } &
		\multicolumn{4}{c}{-4.603226e-05} \\
		\hline
		& MSE & BIAS & $\hat M$ & $\hat N$ & MSE & BIAS & $\hat M$ & $\hat N$ \\
		$\widetilde{\eta}_{n,M,N}$  & 2.40e-07  & 2.79e-05  &  887 & 1  &  1.70e-07  &  4.67e-06  &  2404 & 2  \\
		$\eta^*_{n,M,N}$ & 1.43e-07  &   4.63e-05  &  889 & 1  & 1.49e-07  &  4.76e-06 & 2638  & 1  \\
		\hline
	\end{tabular}
	\caption{Finite sample performance of the FEL $\widetilde{\eta}_{n,M,N}$ and of the estimator $\eta_{n,M,N}^*$. $\overline{\eta}$ represents the average real integrated leverage for each data set. The value of the MSE and BIAS in the table are computed w.r.t. the optimal parameters $\hat{M}$ and $\hat{N}$.}
	\label{control_eff_noise}
\end{table}
Since both estimators are only asymptotically unbiased and in the case of the Heston model the selected cutting frequency parameters $\hat M$ and $\hat N$ are rather small, the variance corrected estimator entails a slight increase of the bias, while for the generalized Heston model the bias remains almost the same.
We notice that in both cases the optimal MSE-based $\hat M$ turns out to be much smaller than the Nyquist frequency (i.e. $\hat M <<n/2$), whereas $\hat N$ is very small, as prescribed by the asymptotic growth conditions in Theorem \ref{T0}.

Let us now analyse the MSE and the sample variance (VAR) of the FEL $\widetilde{\eta}_{n,M,N}$ in the presence of noise for the Heston and the generalized Heston model data sets as a function of $M$ and $N$.

\begin{figure}[h!]
	\centering 
	\includegraphics[width=0.5\textwidth]{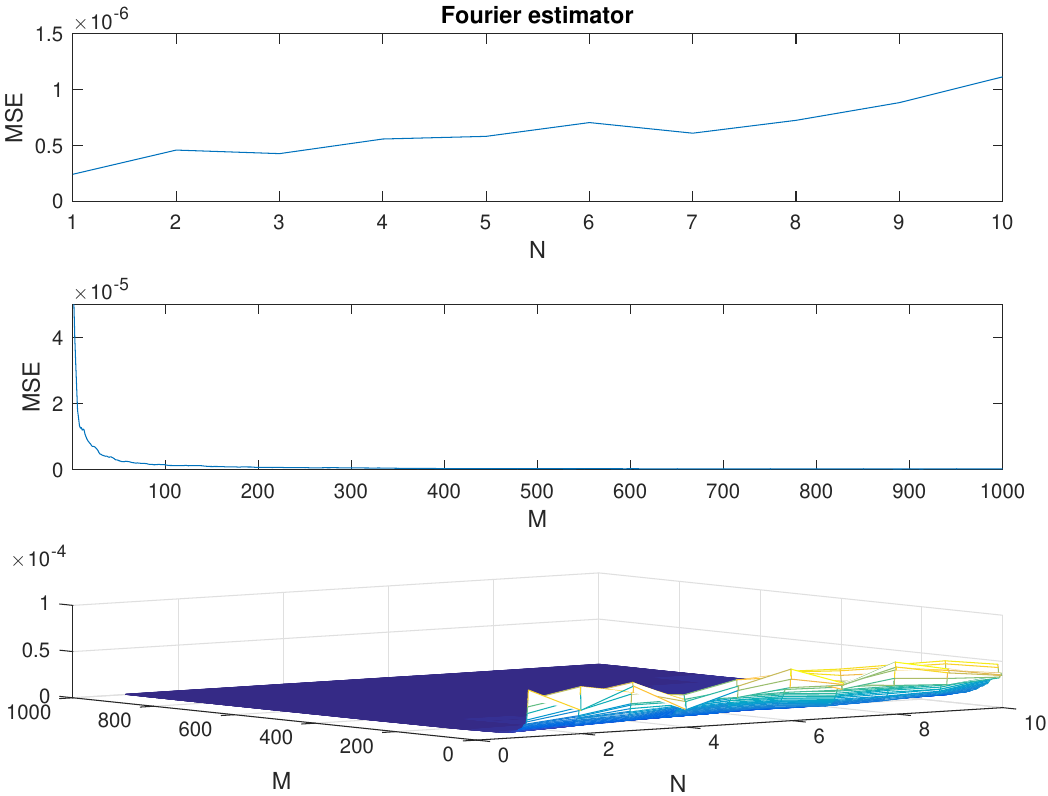}\includegraphics[width=0.5\textwidth]{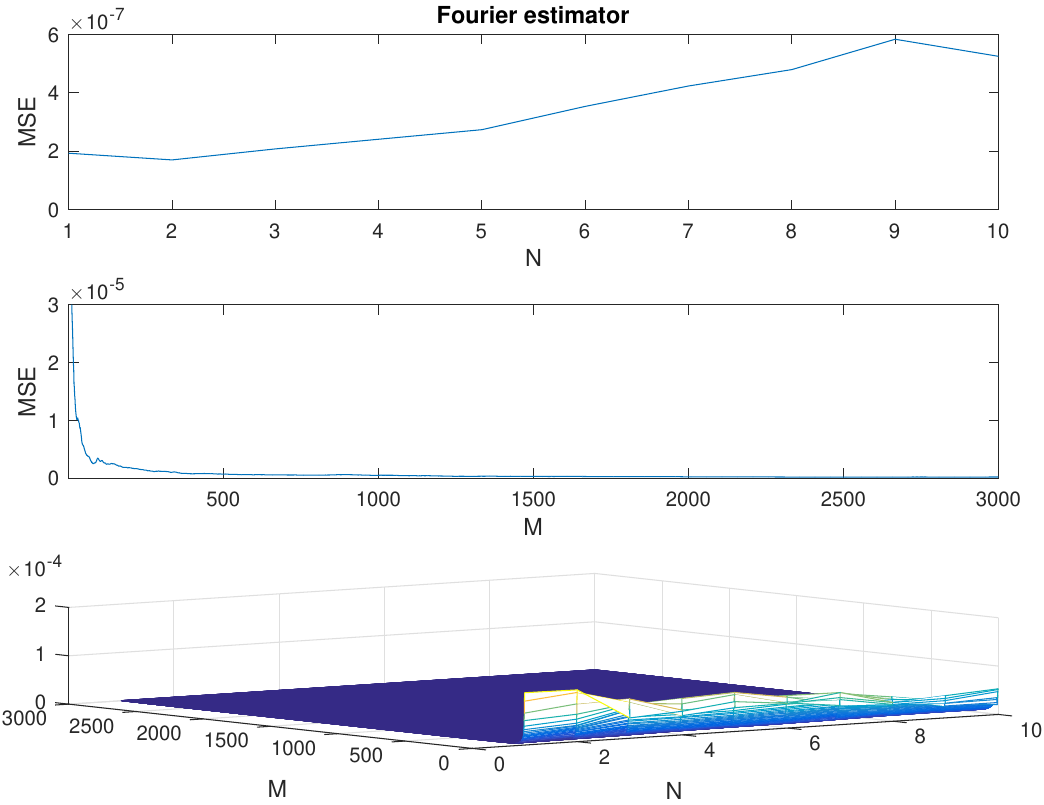}\\
	\includegraphics[width=0.5\textwidth]{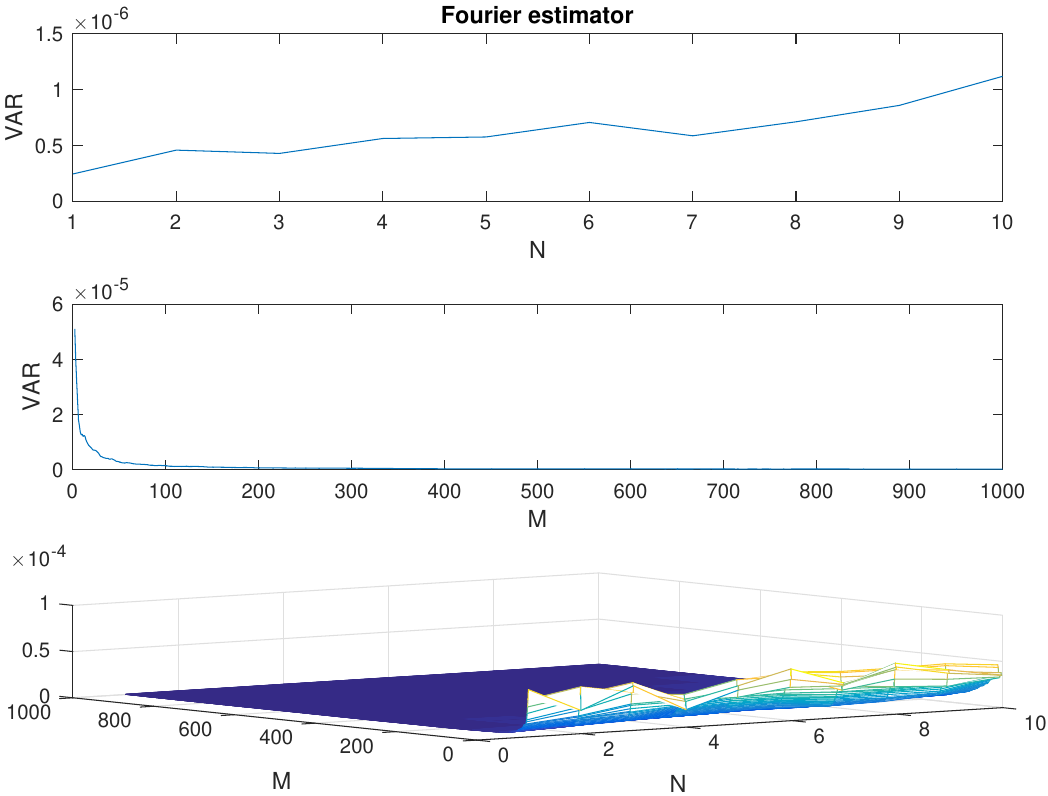}\includegraphics[width=0.5\textwidth]{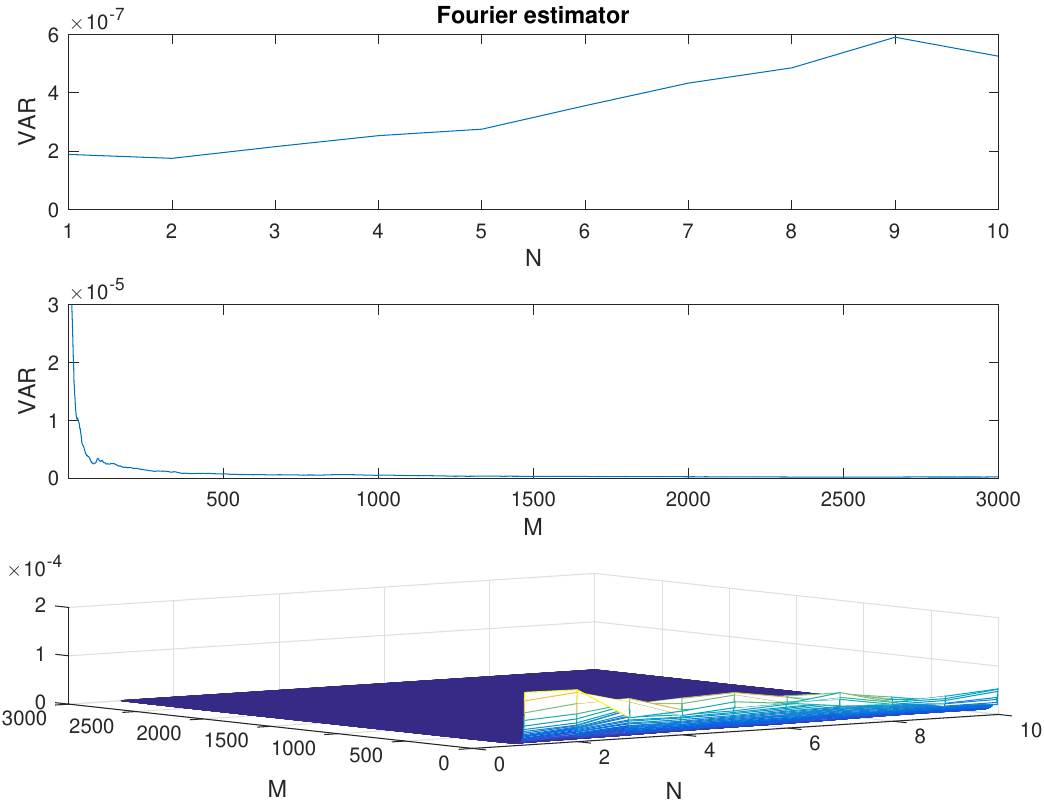}
	\caption{MSE and sample variance of the FEL as a function of $M$ and $N$ under microstructure effects. Left panels: H model. Right panels: GH model.} \label{MSE&VAR}
\end{figure}

From Figure \ref{MSE&VAR}, it is evident that the sample variance of the estimator has the same order of magnitude as the MSE. Moreover, by analysing the relative difference (MSE-VAR)/MSE for the FEL (\ref{decomp_noise}) and the variance corrected estimator (\ref{ControlVar3}) as a function of $M$ and $N$ in Figure \ref{MSE_VARnoise},  we observe that this ratio is negligible for both estimators except for the lowest values of $M$. Moreover, it never exceeds $0.1$ so that the difference MSE-VAR never exceeds $10\%$ of the MSE.
\begin{figure}[h!]
	\centering
	\includegraphics[width=0.5\textwidth]{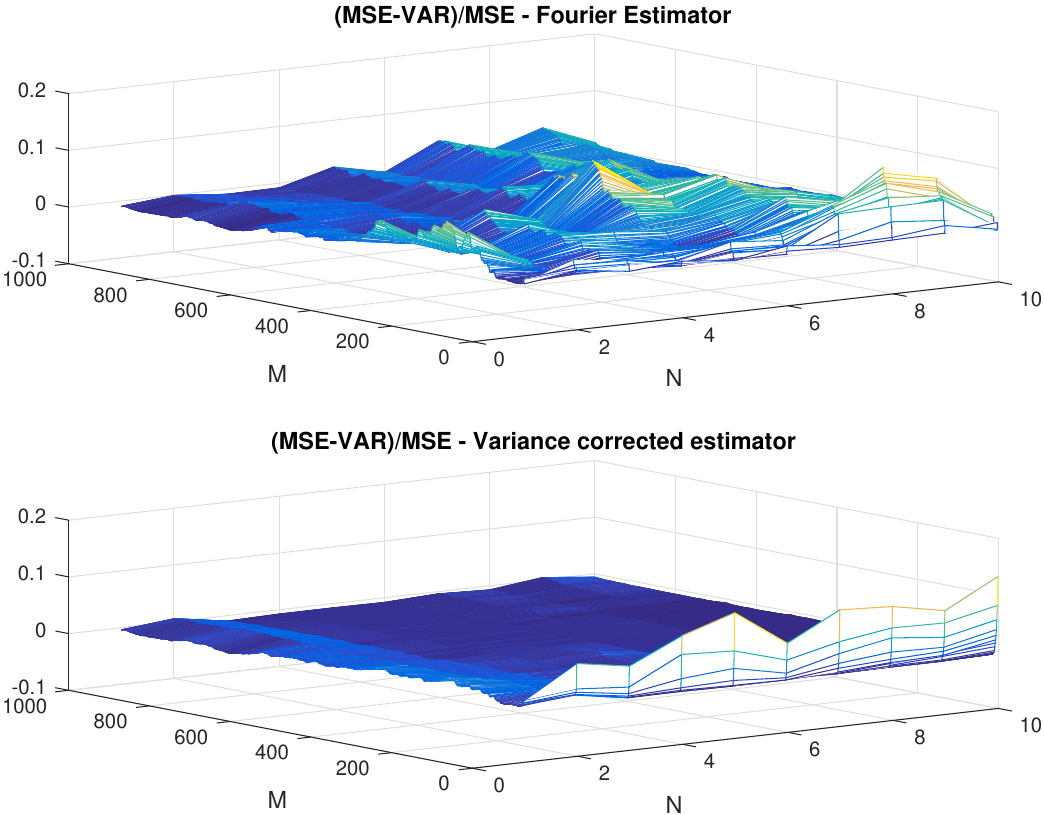}\includegraphics[width=0.5\textwidth]{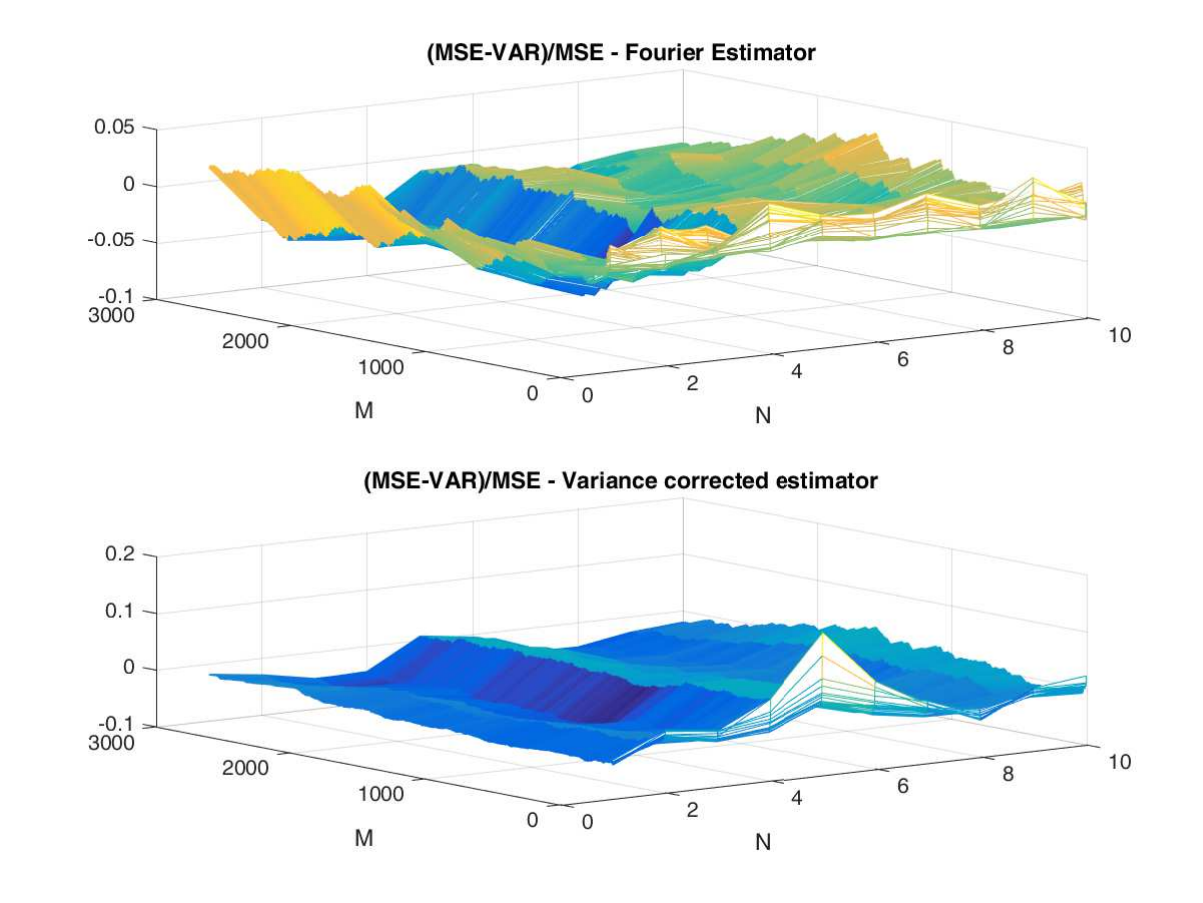}
	\caption{Relative difference between MSE and variance of the FEL $\widetilde{\eta}_{n,M,N}$ and the variance corrected estimator $\eta^*_{n,M,N}$ as a function of $M$ and $N$. Left panels: H model. Right panels: GH model.} \label{MSE_VARnoise}
\end{figure}

We also find that the remaining terms in the MSE decomposition (\ref{mse_decomp}) are at least one order of magnitude smaller than the sample variance, which is then the largest term of the FEL (\ref{decomp_noise}) in the presence of noise. The same conclusions apply when analyzing the variance corrected estimator (\ref{ControlVar3}).
We conclude that minimizing the MSE of the FEL (\ref{decomp_noise}) or of the variance corrected estimator (\ref{ControlVar3}), as a way to determine the optimal cutting frequency parameters $\hat{M}$ and $\hat{N}$, is equivalent to minimizing the sample variance of the estimators.
Following this selection strategy, the parameters $\hat{M}$ and $\hat{N}$ do not yield a minimum value of the estimator bias. The optimized estimator is then affected by a non-negligible bias, which is, however, very small. In our simulation, we always get a significant digit after the comma.
Note that implementing this kind of selection strategy for the variance corrected estimator (\ref{ControlVar3}) means changing Step 4 of its implementation by minimizing the sample variance.
The results obtained by selecting the optimal cutting frequency parameters by minimizing the sample variance of the estimators (\ref{decomp_noise}) and (\ref{ControlVar3}) are displayed in Table \ref{control_eff_noiseVAR}. We highlight that the selected parameters $M$ and $N$ are the same as those selected by MSE minimization. 

\begin{table}
	\centering
	\begin{tabular}{lcccccccc}
		\hline
		& \multicolumn{4}{c}{$H- model$} &
		\multicolumn{4}{c}{$GH- model$} \\
		\hline
		$\overline{\eta}$ & \multicolumn{4}{c}{-1.013673e-04 } &
		\multicolumn{4}{c}{-4.603226e-05} \\
		\hline
		& VAR & $\lambda$ & $\hat M$ & $\hat N$ & VAR & $\lambda$ & $\hat M$ & $\hat N$ \\
		$\widetilde{\eta}_{n,M,N}$ & 2.43e-07  &   &  887 & 1  &  1.75e-07  &    &  2404 & 2  \\
		$\eta^*_{n,M,N}$ & 1.44e-07  &   0.59  &  889 & 1  & 1.51e-07  &  0.86 & 2638  & 1  \\
		\hline
	\end{tabular}
	\caption{Finite sample performance of the FEL $\widetilde{\eta}_{n,M,N}$ and of the estimator $\eta_{n,M,N}^*$. $\overline{\eta}$ represents the average real integrated leverage for each data set. The optimal parameters $\hat{M}$ and $\hat{N}$ are selected by minimization of the sample variance. The value of the VAR in the table is computed w.r.t. the optimal parameters $\hat{M}$ and $\hat{N}$. The symbol $\lambda$ denotes the variance reduction ratio $Var(\eta^*_{n,\hat{M},\hat{N}})/Var(\widetilde{\eta}_{n,\hat{M},\hat{N}})$.}
	\label{control_eff_noiseVAR}
\end{table}

\subsection{Benchmark analysis}

We want now to analyse the FEL performance compared to the estimator of the integrated leverage proposed by \cite{MW} in the presence of noise. The latter is based on pre-averaging and blocking that allows us to deal with the noise contained in the data. Here two nested levels of blocks are required: the first one, of size $M$, defines the range of pre-averaging, and the second one, of size $L$, is used for computing the realized covariance between returns and volatility increments. 
Our choice for the blocking parameters $M$ and $L$ is the following: we let $M$ vary from 2 seconds to 300 seconds (i.e. 5-minute block size). Then, up to rounding, for each value of $M$ we define $n'=n/M$ and let $L=[\sqrt{n'}]$. Coherently with our previous approach, we then choose the optimal parameters by directly minimizing the MSE over the range of $M$'s. We call this estimator WM1.

In their paper, the authors provide a rule to choose the optimal values of $M$ and $L$ that minimize the asymptotic variance in the presence of microstructure effects. However, the implementation of this rule requires a preliminary estimate of the integrated volatility, the integrated quarticity and the integrated sixth power of volatility, besides the estimation of the spot quarticity and the diffusion coefficient $\gamma(t)$ in (\ref{mod}). To reduce possible sources of estimation errors, we compute these quantities from the model (\ref{HestonModel}) by Riemann integration rule. We call this estimator WM2. Our results are resumed in Table \ref{confronto}.
We notice that the first procedure, which is completely unfeasible, provides a worse estimate than the Fourier methodology presented in Table \ref{control_eff_noise} and \ref{control_eff_noiseVAR} both in terms of bias and variance. On the other hand, the second procedure provides a very good estimate in terms of bias, while the variance and MSE are nevertheless slightly larger than those obtainable by the Fourier approach. This does not come as a surprise, since the estimator proposed by \cite{MW} contains a bias correction factor while the Fourier estimator achieves unbiasedness only asymptotically.
\begin{table}[h!]
	\centering
	\begin{tabular}{lccccc}
		\hline
		\cite{MW} & \multicolumn{5}{c}{$H- model$}  \\
		\hline
		Estimator & MSE & VAR & BIAS & $M$ & $L$ \\
		WM1 & 3.90e-06 & 3.76e-06  & -4.24e-04  &  4 & 73  \\
		WM2 & 3.09e-07 & 3.12e-07  &   -7.96e-06  &  2 & 2460   \\
		\hline
	\end{tabular}
	\caption{Finite sample properties of the estimator by \cite{MW} under microstructure effects.}
	\label{confronto}
\end{table}
As a further evidence, in Figure \ref{MW2014} we can see that the estimate proposed by \cite{MW} is largely dependent on the choice of the block size $M$ and its MSE is increasing with this parameter, while the bias remains rather stable around zero.
\begin{figure}[h!]
	\centering 
	\includegraphics[width=0.5\textwidth]{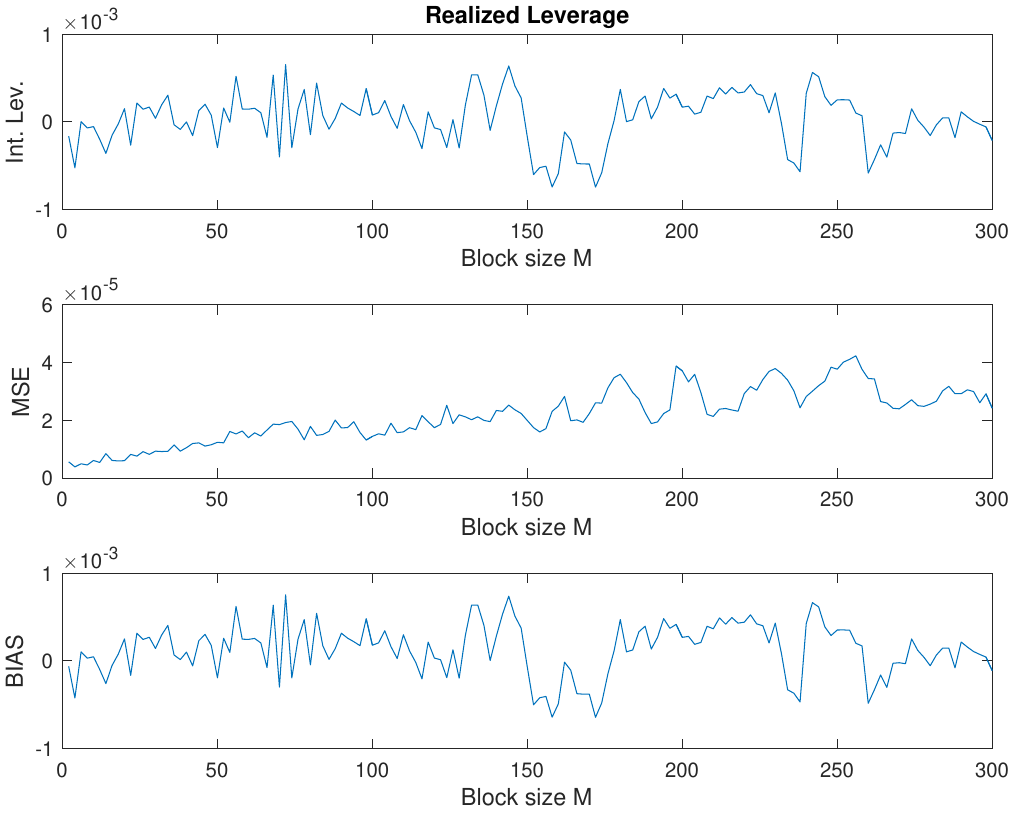}
	\caption{MSE-based integrated leverage estimate by \cite{MW} together with its MSE and BIAS as a function of the block size $M$.} \label{MW2014}
\end{figure}

\subsection{Sensitivity analysis on the generalized Heston model}

We examine the FEL behaviour in the presence of noise depending on the choice of some parameters of the GH model.
The process $\rho(t)$ in the GH model is a linear transformation of a Jacobi process which takes values in $[-1,1]$. Moreover, $\rho(t)$ is mean reverting to $\zeta=(2\xi-\eta)/\eta$ at speed $\eta$. These kinds of processes are ideal diffusions to model stochastic correlation. Its properties are summarized by \cite{VV}.

An interesting feature of this process is that it tends to a jump process with state-space $\{-1,1\}$ and constant intensities if $\theta$ tends to infinity. Roughly speaking, when $\theta$ increases the process exhibits a jump-type behaviour while the smaller the parameter $\theta$, the smoother are the sample paths.
Therefore, the fluctuations of the process $\rho(t)$ can be amplified by increasing the parameter $\theta$.
Fig. \ref{Sensistivity_theta} shows the sensitivity of the FEL (\ref{decomp_noise}) and of the variance corrected estimator (\ref{ControlVar3}) in terms of MSE with respect to the choice of $\theta$. All the other model parameters are set as in Section \ref{sec5.1} and the cutting frequency parameters are determined by minimization of the sample variance.  As expected, the MSE of both estimators slightly deteriorates as the jump-type behaviour of the correlation process is emphasized.
\begin{figure}[ht]
	\centering 	\includegraphics[width=0.5\textwidth]{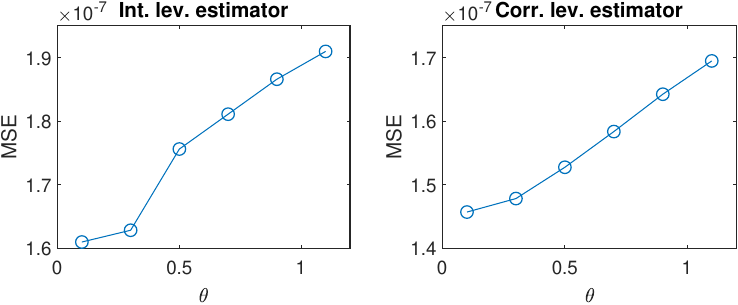}
	\caption{Sensitivity of the FEL $\widetilde{\eta}_{n,M,N}$ and of the estimator $\eta_{n,M,N}^*$ with respect to the choice of $\theta$ in the presence of microstructure noise. Parameter values: $\alpha=0.01, \beta=0.2,\nu=0.05$ and $\xi=0.02$, $\eta=0.5$.} \label{Sensistivity_theta}
\end{figure}

\begin{figure}[ht!]
	\centering 	\includegraphics[width=0.5\textwidth]{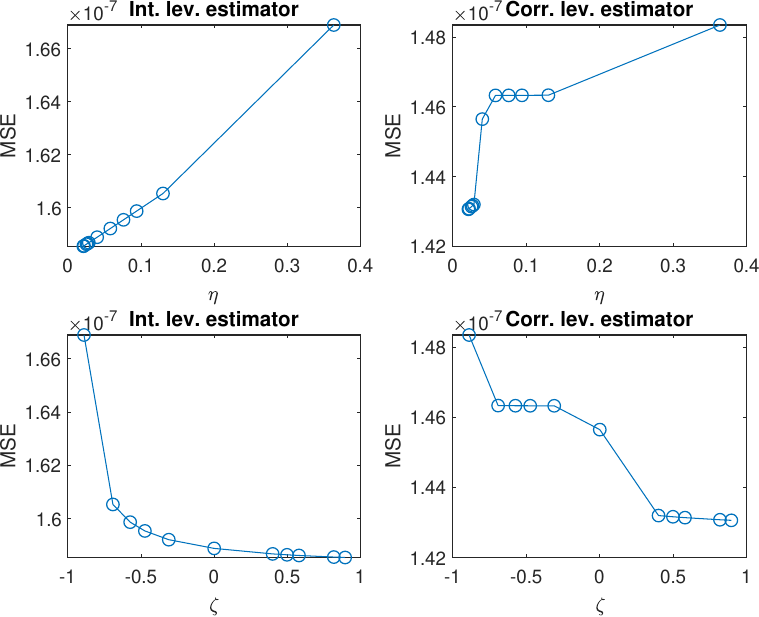}
	\caption{Sensitivity of the FEL $\widetilde{\eta}_{n,M,N}$ and of the estimator $\eta_{n,M,N}^*$ with respect to the choice of $\eta$ and $\zeta=(2\xi-\eta)/\eta$ in the presence of microstructure noise. Parameter values: $\alpha=0.01, \beta=0.2,\nu=0.05$ and $\xi=0.02$, $\theta=0.5$.}
	\label{Sensistivity_eta} 
\end{figure}

We also examine the sensitivity of the FEL to the mean reversion parameter $\eta$. The effect is examined in Fig. \ref{Sensistivity_eta}, where the MSE is plotted as a function of $\zeta$ ranging in $[-1,1]$ and of the corresponding $\eta=2\xi/(\zeta+1)$. In this case, the variability of the MSE is small. 
Moreover, we highlight that the Fourier estimator is not affected much by the speed of mean reversion and performs slightly better when the speed of mean reversion is lower. That makes the Fourier methodology particularly suitable to apply in a general setting where we can assume that the GH model is the data generating process.

\section{Estimation of the stochastic leverage effect in the presence of microstructure noise}
\label{sec6}

In the following, we denote by $\tilde\sigma_{n,M}^2$ and $\tilde \gamma_{n,M,N}^2$ the FEV and the FEVV in the presence of microstructure noise, respectively. We can then obtain a plug-in estimator of (\ref{rho}) by using the FEL defined in (\ref{decomp_noise}) at the numerator and $\tilde\sigma_{n,M}^2$ and $\tilde \gamma_{n,M,N}^2$ at the denominator
\begin{equation}
	\label{R_T}
	\tilde R_T= \frac{\tilde \eta_{n,M,N}}{\sqrt{\tilde\sigma_{n,M}^2 \tilde \gamma_{n,M,N}^2}}.
\end{equation}

Similarly, we can also consider a second estimator of $(\ref{rho})$ by using the Fourier estimator $\eta^*_{n,M,N}$ defined in (\ref{ControlVar3}) at the numerator
\begin{equation}
	\label{R_T*}
	R_T^*= \frac{\eta^*_{n,M,N}}{\sqrt{\tilde\sigma_{n,M}^2 \tilde \gamma_{n,M,N}^2}}.
\end{equation}

The estimators (\ref{R_T}) and (\ref{R_T*}) depend on the choice of the cutting frequency parameters $M$ and $N$ that strongly affect the quality of the estimates. 

For the FEL at the numerator, we follow the selection strategy for the parameters $M$ and $N$ described in Section \ref{sec5.1}. If we have simulated data, we generate a certain number of daily samples and minimize the sample variance of the FEL. On the other hand, if we are working in a real data framework, we need several days of observations (e.g. $100$ days) to perform the parameter selection.

For the FEV $\tilde\sigma_{n,M}^2$ the cutting frequency is determined by minimizing the MSE estimate determined in \cite[Theorem 3]{MS08} on each day. Finally, for the FEVV $\tilde \gamma_{n,M,N}^2$ the frequencies $M$ and $N$ can be chosen in the range defined  in \cite[Remark 4.3]{CMS15}. 

In this section, we use simulated data while performing an estimation with real data in Section \ref{sec7}.
We again simulate second-by-second return and variance paths over a daily trading period of $T=6$ hours from the GH model (\ref{GHmodel}) and choose the parameter values and the noise-to-signal ratio as in Section \ref{sec5.1}.  In Table \ref{cut-off}, we list the value of the cutting frequency parameters selected for the simulation, together with the MSE achieved by each estimator appearing in (\ref{R_T})-(\ref{R_T*}) and the average value over $100$ days of the corresponding true (following the GH model) integrated leverage, integrated volatility and integrated volatility of volatility.
\begin{table}[ht]
	\centering
	\begin{tabular}{lcccc}
		\hline
		Estimates & Reference Value	& MSE & $M$ & $N$  \\
		\hline
		$\widetilde{\eta}_{n,M,N}$  & -4.60e-05 & 1.70e-07 & 2404 & 2 \\
		$\eta^*_{n,M,N}$ & -4.60e-05 & 1.49e-07 & 2638 & 1 \\
		$\tilde\sigma_{n,M}^2$  & 1.00e-02 &  2.38e-06 & 819 & -  \\
		$\tilde \gamma_{n,M,N}^2$ & 2.51e-05 &  3.71e-07 & 146 & 3   \\
		\hline
	\end{tabular}
	\caption{GH data set. MSE and parameters' selection related to the estimators appearing in (\ref{R_T}) and (\ref{R_T*}). The reference values corresponds to the average over $100$ days of the corresponding true integrated values.}
	\label{cut-off}
\end{table}

Figure \ref{FigR_T} shows $R_T$ estimated using (\ref{R_T}) and (\ref{R_T*}). The true $R_T$ is plotted in red. In our simulation, both plots displayed in Figure \ref{FigR_T} seem to provide a good approximation of $R_T$.
\begin{figure}[ht]
	\centering 	\includegraphics[width=0.5\textwidth]{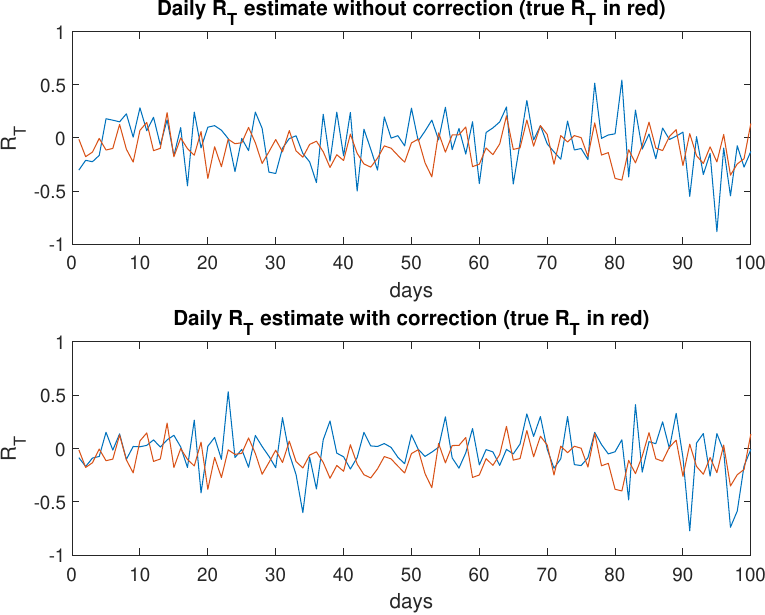}
	\caption{Upper panel: Daily $R_T$ estimates performed with $\tilde R_T$. Lower panel: Daily $R_T$ estimates performed with $ R_T^*$. The true values of $R_T$ are plotted in red.} 
	\label{FigR_T}\end{figure}

\section{Empirical analysis}
\label{sec7}

We analyse the leverage effect pattern in a tick data set by using the estimators of the stochastic leverage effect $\tilde R_T$ and $R_T^*$ defined in (\ref{R_T}) and (\ref{R_T*}), respectively.

We consider transaction data of the S\&P500 futures recorded at the Chicago Mercantile Exchange (CME) for the period from January 3, 2007, to December 31, 2008 (502 days).
During this period, the United States experienced the subprime mortgage crisis, a nationwide financial crisis that contributed to the U.S. recession of December 2007 till June 2009. It was triggered by a large decline in home prices after the collapse of a housing bubble during 2006. That induced a large banking crisis in 2007 and the financial crisis in 2008. In nine days from October 1 to 9, 2008 the S\&P500 lost 21.6\% of its value.
Table  \ref{dati} describes the main features of our data set.
\begin{table}[ht]
	\centering
	\begin{tabular}{lclcccc}
		\hline
		Year & N. trades & Variable & Mean & Std. Dev. & Min & Max \\
		\hline
		2007 & 566409 & S\&P 500 index    &  1484.84 & 44.30 & 1375.00 & 1586.50   \\
		& & log-return  & 5.00e-6 & 1.81e-2 & -1.64  & 2.33 \\
		2008 & 557982 & S\&P 500 index    &  1226.55 & 186.89 & 739.00 & 1480.20   \\
		& & log-return  & -9.03e-5 & 4.75e-2 & -8.66  & 6.12 \\
		\hline
	\end{tabular}
	\caption{Summary statistics for the sample of the traded CME
		S\&P500 futures for the period from January 3rd 2007 to December 31st 2008 (502 days).}
	\label{dati}
\end{table}

Figure \ref{SPdata} shows the plot of the log-prices and returns for the raw transaction data. High-frequency returns are contaminated by microstructure effects, such as transaction costs and bid-and-ask bounce effects,  leading to biases in the variance measures.
Figure \ref{ACF} shows the autocorrelation function for the log-returns. Raw data exhibit a strongly significant positive first-order autocorrelation and higher-order autocorrelations remain significant up to lag 8 in 2007 and up to lag 15 in 2008.
\begin{figure}
	\includegraphics[width=0.5\textwidth]{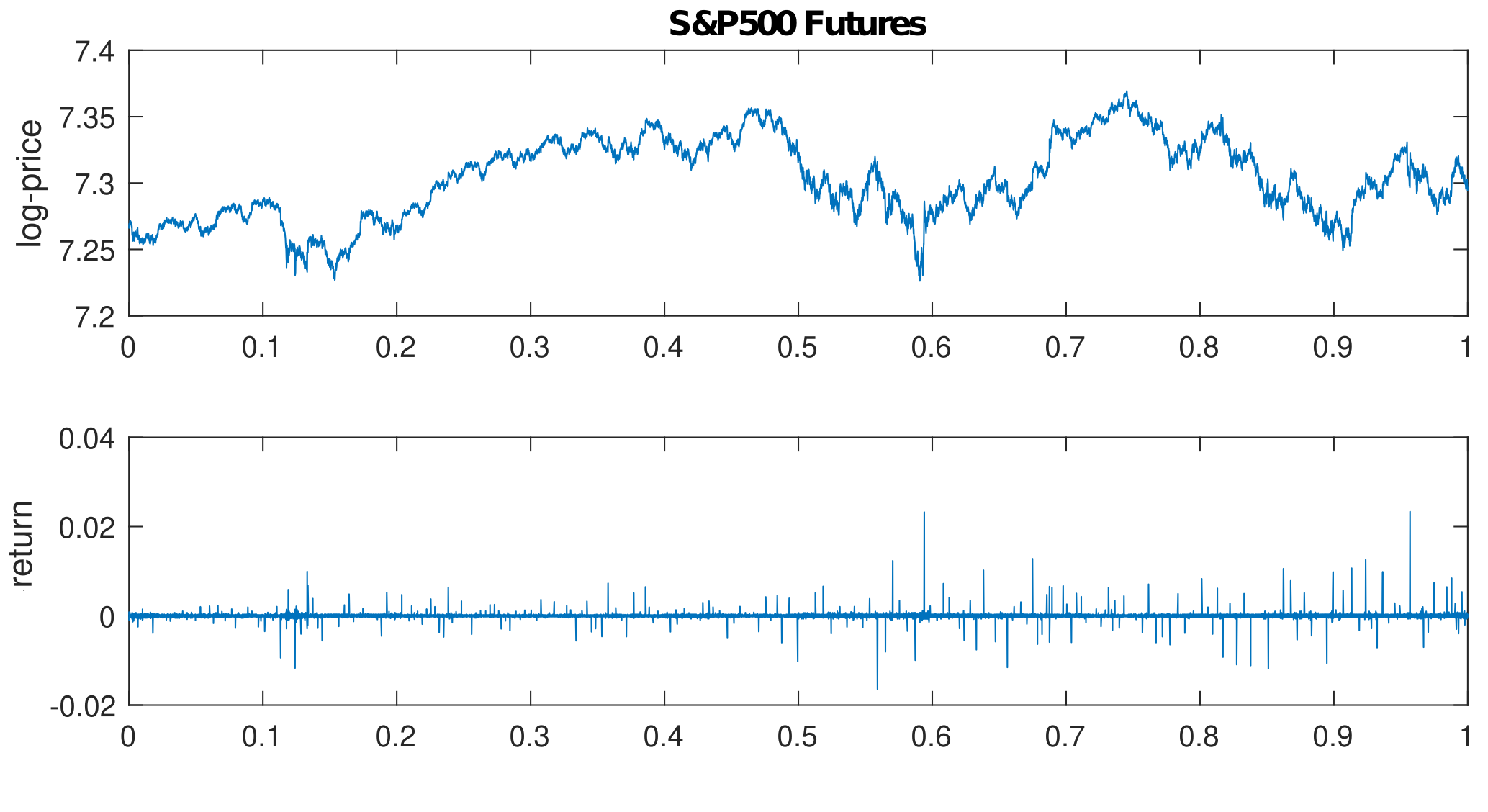}\includegraphics[width=0.5\textwidth]{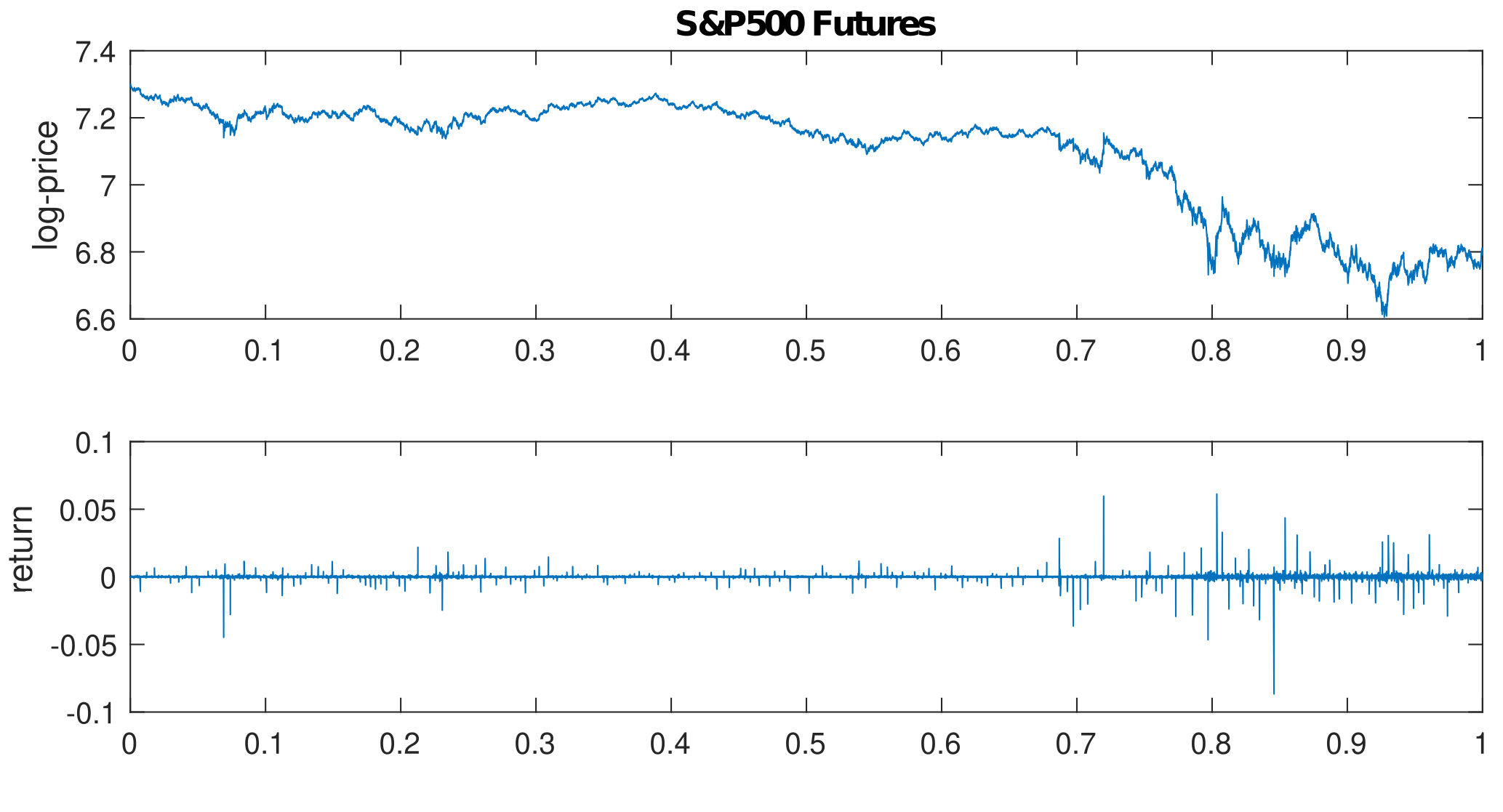}
	\caption{Log-prices and returns plots for S\&P500 futures in the years 2007 (left panels) and 2008 (right panels).} \label{SPdata}
\end{figure}
\begin{figure}
	\includegraphics[width=0.5\textwidth]{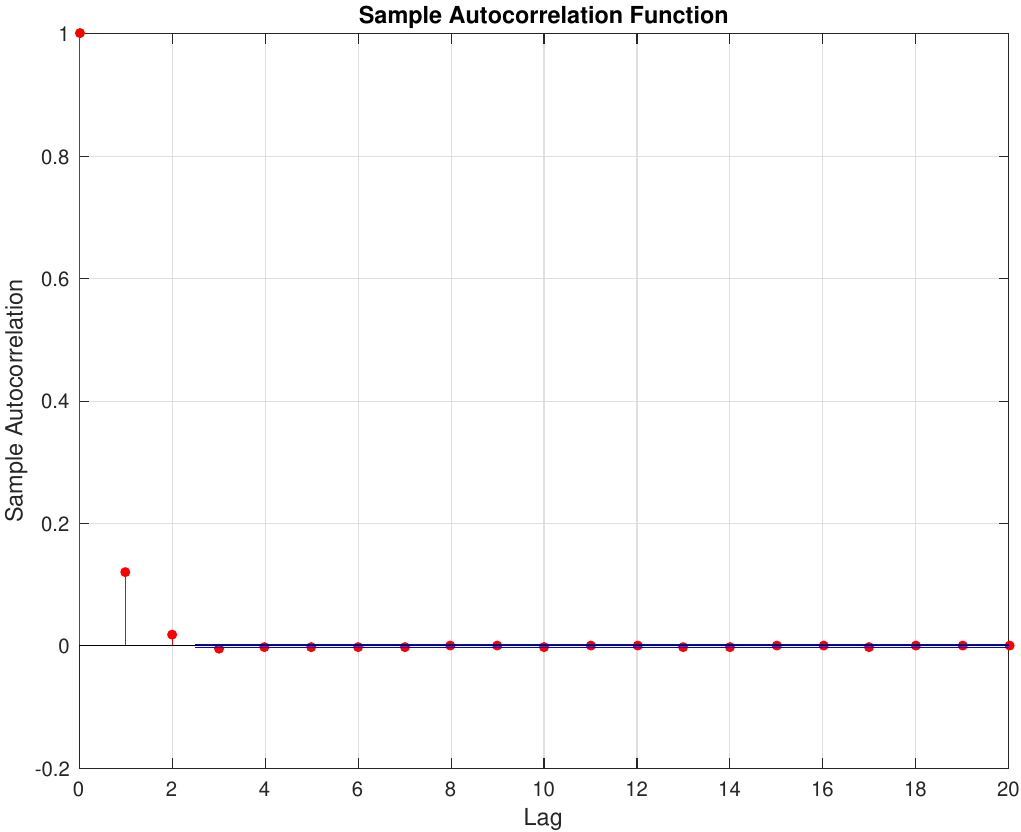}\includegraphics[width=0.5\textwidth]{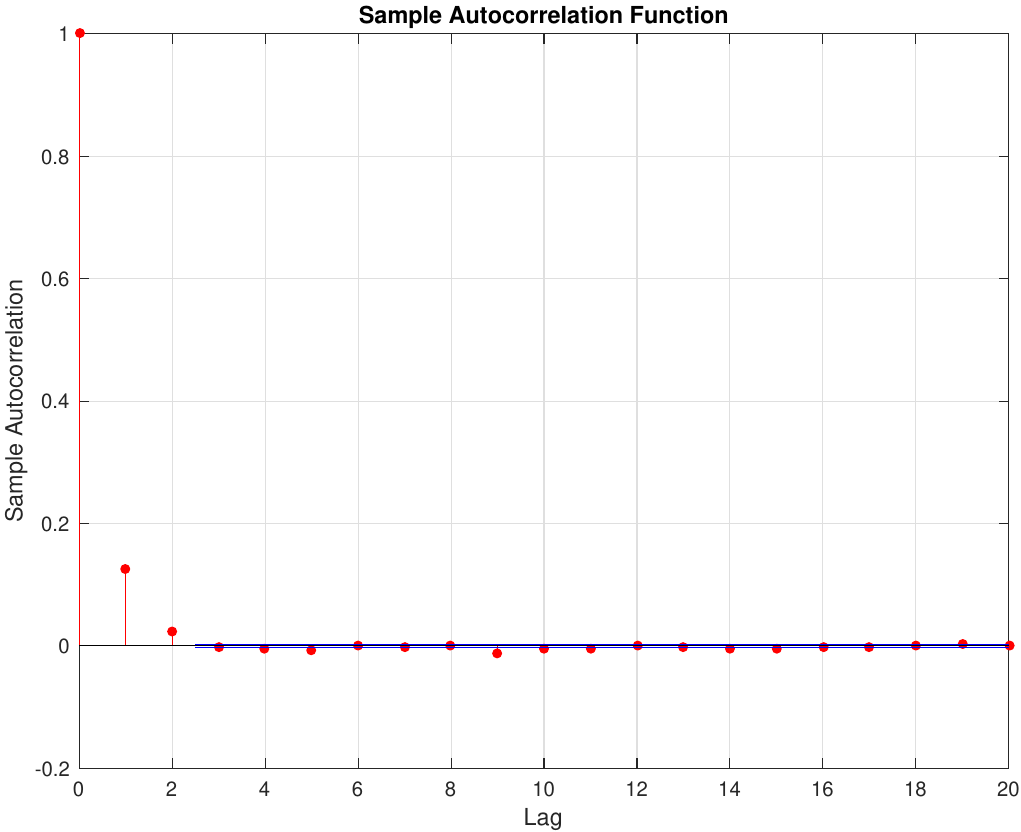}
	\caption{Autocorrelation function for S\&P500 futures in the years 2007 (left panel) and 2008 (right panel).} \label{ACF}
\end{figure}
We then perform an analysis of the stochastic leverage effect using the estimators $\widetilde{R}_T$ and $R_T^*$.
First, we plot the daily values of the factors appearing in the estimator $\widetilde{R}_T$ and $R_T^*$.
The upper panels of Figure \ref{RealLEV} show the daily FEL (\ref{decomp_noise}) and its variance corrected version (\ref{ControlVar3}) in 2007 (left) and 2008 (right). The middle panels show the FEV together with the 5-minute sparse sampled realized volatility estimator (which we consider as a benchmark of our estimates) and the lower panels show the FEVV in 2007 and 2008. The estimated quantities in 2007 and 2008 are very different in magnitude. The year 2008 displays the largest values (both negative and positive) of all the metrics, coherently with the occurrence of the financial crisis. During 2007 the integrated leverage is rather small and mostly negative.
All the estimations are almost flat during 2008 up to September 16 (day 177 in our sample), when the integrated leverage exhibits the first large negative spike. The second negative spike is on September 29 (day 186), which corresponds to the beginning of the financial crisis.
Our finding highlights the presence of persistent positive and negative integrated leverage, especially in periods of financial turmoil. We notice that both the FEL (\ref{decomp_noise}) and the variance corrected estimator (\ref{ControlVar3}) catch the same positive and negative spikes of the integrated leverage; nevertheless, the estimator (\ref{ControlVar3}) exhibits a smaller variability. 
\begin{figure}
	\includegraphics[width=0.5\textwidth]{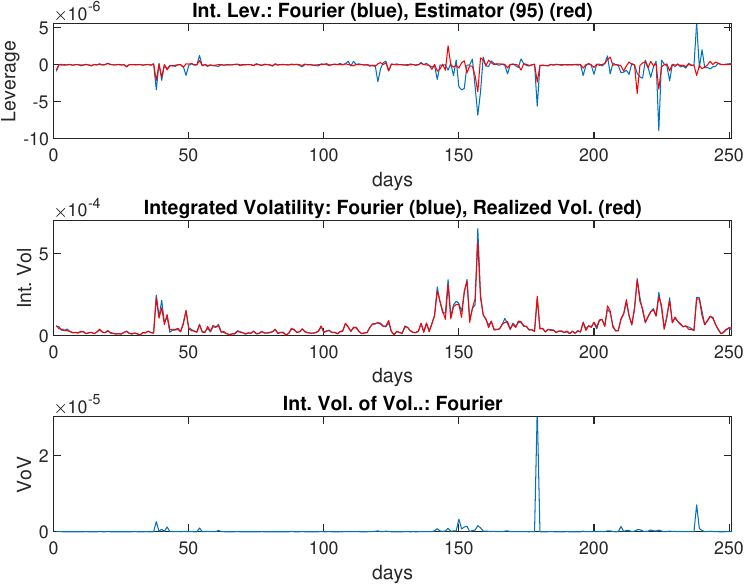}\includegraphics[width=0.5\textwidth]{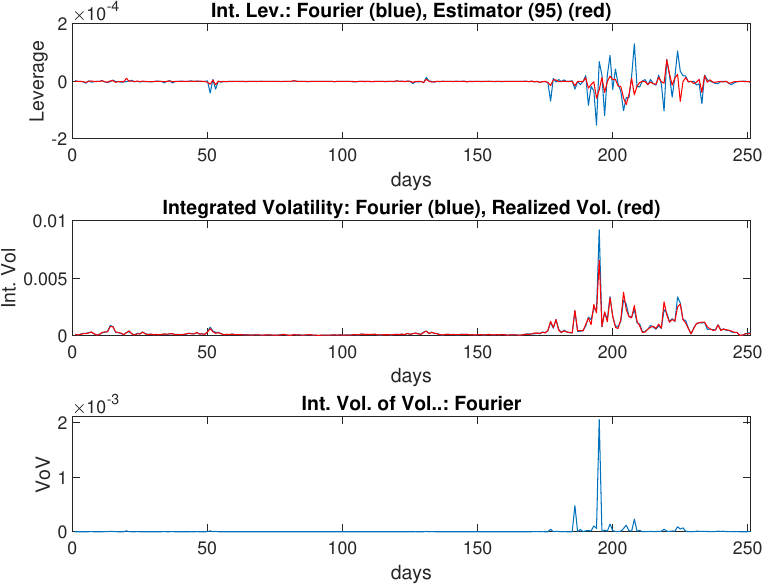}
	\caption{Upper panels: FEL (\ref{decomp_noise}) (blue) and variance corrected estimator (\ref{ControlVar3}) (red) . Middle panels: FEV (blue) and the realized volatility estimator (red) in the years 2007 (left panel) and 2008 (right panel). Lower panels: FEVV (blue) in the years 2007 (left panel) and 2008 (right panel).
	}
	\label{RealLEV}
\end{figure}

\begin{table}
	\centering 
	\begin{tabular}{lccccc}
		\hline
		{S\&P500 futures}	& \multicolumn{5}{c}{$2007$} \\
		\hline
		& Estimate & VAR & $\lambda$ & $\hat M$ & $\hat N$  \\
		$\widetilde{\eta}_{n,M,N}$ & -2.44e-07 & 1.20e-12  &   &  363 & 1   \\
		$\eta_{n,M,N}^*$ & -9.23e-08 & 3.02e-13  &   0.25  &  143 & 1  \\
		& & & & & \\
		\hline
		{S\&P500 futures}	& \multicolumn{5}{c}{$2008$} \\
		\hline
		& Estimate & VAR & $\lambda$ & $\hat M$ & $\hat N$\\
		$\widetilde{\eta}_{n,M,N}$	& -2.01e-06 & 5.94e-10  &    &  281 & 3 \\
		$\eta_{n,M,N}^*$ & -1.76e-06 & 1.51e-10  &  0.25 & 285  & 3 \\
		\hline
	\end{tabular}
	
	\caption{The FEL, its variance corrected counterpart (\ref{ControlVar3}) and their sample variance computed w.r.t. the optimal parameters $\hat{M}$ and $\hat{N}$. The symbol $\lambda$ denotes the variance reduction ratio $Var(\eta^*_{n,\hat{M},\hat{N}})/Var(\widetilde{\eta}_{n,\hat{M},\hat{N}})$. The optimal cutting frequency parameters are obtained by minimization of the sample variance over each year. The estimates in the first column correspond to averages over all the year. }
	\label{SP_MN}
\end{table}

When estimating the integrated leverage, a larger variability can be observed than estimating other quantities such as volatility or quarticity. According to the analysis of Section \ref{sec6}, for both estimators $\widetilde{\eta}_{n,M,N}$ and $\eta^*_{n,M,N}$ the cutting frequency parameters $M$ and $N$ are chosen such to minimize the sample variance over the whole one year sample. Their optimal values are listed in Table \ref{SP_MN}, along with the sample variance achieved by the FEL (\ref{decomp_noise}) and its variance corrected counterpart (\ref{ControlVar3}).
Due to the presence of microstructure effects, the optimal cutting frequency $\hat M$ turns out to be much smaller than the Nyquist frequency (i.e. $M \ll n/2=2460$).

\begin{figure}
	\includegraphics[width=0.5\textwidth]{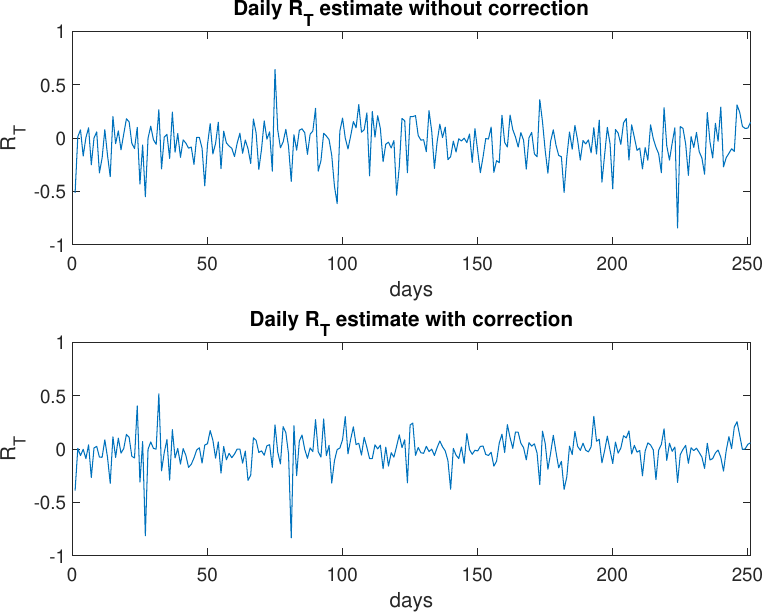}\includegraphics[width=0.5\textwidth]{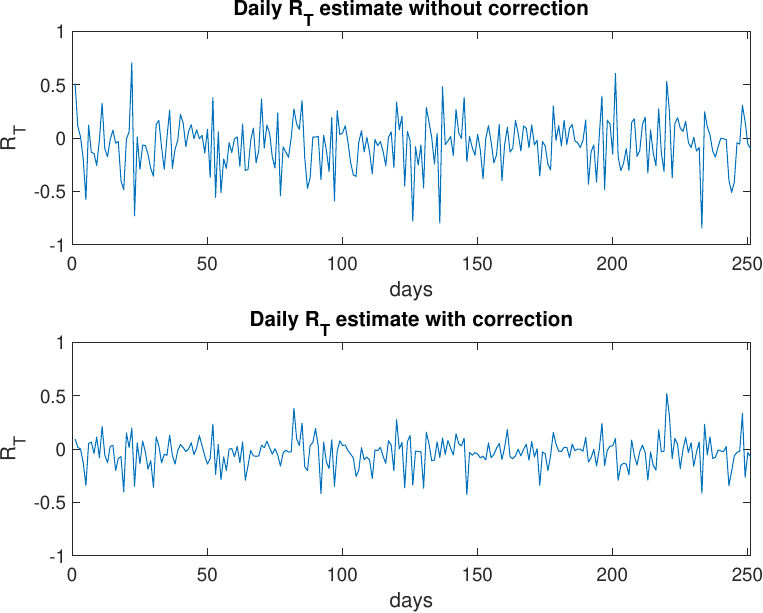}
	\caption{$R_T$ estimated by the FESL (\ref{R_T}) (upper panels) and by its variance corrected version (\ref{R_T*}) (lower panels) in the years 2007 (left panels) and 2008 (right panels).
	}
	\label{R_T_SP}
\end{figure}

We highlight that the Fourier estimator makes use of all the $n$ observed prices and it filters out microstructure effects by a suitable choice of $M$ and $N$, instead of reducing the sampling frequency.

We conclude this section by showing $R_T$ estimates obtained with the estimators (\ref{R_T}) and (\ref{R_T*}). The graphs do not show evident differences between the years 2007 and 2008. That is how it should be. Indeed, we expect to see a similar behaviour of the stochastic leverage effect estimates because we assume that the data generating process is the same for both the 2007 and 2008 data sets. The obtained estimates seem to express a \emph{fundamental measure} of the asymmetry between returns and volatilities observed across the years.

\section{Conclusions}
\label{sec8}
We define an estimator of the stochastic leverage effect which combines the Fourier estimator of the integrated leverage (FEL), of the integrated volatility (FEV) and the integrated volatility of volatility (FEVV) defined in \cite{CS15}, \cite{MM} and \cite{CMS15}, respectively. We call it the \emph{Fourier estimator of the stochastic leverage effect} (FESL). An advantage of this estimator is that it avoids estimating the latent volatility path. This step is mandatory in concurrently realized covariance-based estimators appearing in the literature and is one reason behind several bias corrections required by them.

We show the consistency of the FESL in the absence of microstructure noise. Then, we focus on analyzing the behaviour of the FESL in the presence of microstructure noise. The latter is strictly related to the finite sample properties of the FEL, the FEV, and the FEVV used in the estimation. The only estimator for which a thorough analysis of the latter is missing in the literature is the FEL. We fill this gap and determine that it is asymptotically unbiased but has a diverging mean squared error. 

We propose a variance corrected estimator of the FEL to hinder its variability in the finite sample. We then examine selection strategies for the cutting frequency parameters appearing in the estimation methodology. Numerically, we observe that the parameter values obtained by minimizing the mean squared error of the FEL are equivalent to those obtained minimizing its sample variance. Moreover, we note that using the variance corrected estimator reduces the sample variance of the final estimation by a half and that a selection strategy based on its sample variance is also directly applicable when real data are at disposal. Finally, we investigate the performance of the FESL in empirical analysis and detect the presence of the stochastic leverage effect in the S\&P 500 future prices data set for the years 2007 and 2008.

\section*{Appendix: Proofs}

In the proofs below, we make often use of following Lemma, see \cite{Katz}.
\begin{Lemma}
	\label{cor_dir}
	Let $D_N(t)$ be the normalized Dirichlet kernel defined in (\ref{dir}), then the following properties are satisfied.
	\begin{enumerate}
		\item $ \int_0^{T} |D_N(u)|^2 \, du = \frac{T}{2N+1} $,
		\item  $\forall p > 1$, there exists a constant $\mathcal{C}_p$ such that $\int_0^{T} |D_N(u)|^p \, du = \frac{\mathcal{C}_p}{2N+1} $.
	\end{enumerate}
\end{Lemma}

\begin{proof}[Proof of Theorem \ref{T0}:]
	
	Throughout the proof we indicate with $\phi_n(s):=\sup_{k=0,\ldots,n}\{t_{k} : t_k \leq s \}$.
	Moreover, we use the following integral notation for the discrete Fourier coefficients of the returns (\ref{d1})
	\[
	c_n(s;dp)=\frac{1}{T} \int_{0}^{T} \mathrm{e}^{-\mathrm{i}\frac{2\pi}{T}s\phi_n(u)} dp(u).
	\]
	In the proof, $C$ will denote a positive constant, not necessarily the same at different occurrences.
	
	Applying the product rule to the term $c_n(s;dp)c_n(l-s;dp)$ and using notations (\ref{dir}) and (\ref{der}), we obtain the following error decomposition.
	\begin{align}
		&\eta_{n,M,N}-\eta=\eta_{n,N}^1 -\eta+\eta_{n,M,N}^2  \label{est_error}\\
		&= \underset{\eta_{n,N}^1}{\underbrace{\frac{T^2}{2N+1} \sum_{|l|\leq N} \mathrm{i}l \frac{2\pi}{T} \frac{1}{T}   \int_{0}^{T} \mathrm{e}^{-\mathrm{i}\frac{2\pi}{T}l \phi_n(t)} \sigma^2(t) dt \frac{1}{T} \int_0^T  \mathrm{e}^{\mathrm{i}\frac{2\pi}{T}l\phi_n(u)} dp(u)}}- \int_0^T \eta(t) \, dt \label{p1}\\
		&+ \underbrace{\frac{T^2}{2N+1} \sum_{|l|\leq N} \mathrm{i}l \frac{2\pi}{T} \Bigg(\frac{1}{T}\int_{0}^{T} \int_{0}^{t} \mathrm{e}^{-\mathrm{i}\frac{2\pi}{T}l\phi_n(u)} D_M(\phi_n(t)-\phi_n(u)) \, dp(u) \,dp(t)} \nonumber
	\end{align}
	\begin{align}
		&+ \underset{\eta_{n,M,N}^2}{\underbrace{\frac{1}{T}\int_{0}^{T} \mathrm{e}^{-\mathrm{i}\frac{2\pi}{T} l\phi_n(t)} \int_{0}^{t}  D_M(\phi_n(t)-\phi_n(u)) \, dp(u) \,dp(t) \Bigg) \, \frac{1}{T} \int_0^T  \mathrm{e}^{\mathrm{i}\frac{2\pi}{T}l\phi_n(u) }, dp(u) }}, \nonumber
	\end{align}
	where the variable $t \in (0,T]$.
	By using the Cauchy-Schwartz inequality we have that
	
	\begin{align}
		&\E[|\eta_{n,M,N}^2|] \leq \frac{T^2}{2N+1} \sum_{|l|\leq N} |l| \frac{2\pi}{T} \E \Big[ \Big| \frac{1}{T}\int_{0}^{T} \int_{0}^{t} \mathrm{e}^{-\mathrm{i}\frac{2\pi}{T}l\phi_n(u)} D_M(\phi_n(t)-\phi_n(u)) \, dp(u) \,dp(t) \nonumber\\
		&+ \frac{1}{T}\int_{0}^{T} \mathrm{e}^{-\mathrm{i}\frac{2\pi}{T} l\phi_n(t)} \int_{0}^{t}  D_M(\phi_n(t)-\phi_n(u)) \, dp(u) \,dp(t) \Big|^2 \Big]^{\frac{1}{2}} \E\Big[\Big| \frac{1}{T} \int_0^T  \mathrm{e}^{\mathrm{i}\frac{2\pi}{T}l\phi_n(u)} dp(u) \Big|^2 \Big]^{\frac{1}{2}}. \nonumber
	\end{align}
	
	For each $|l|\leq N$, the $L_2$-norm of the Fourier coefficients of the returns
	\[
	\E\Big[ \Big |\int_0^T  \mathrm{e}^{\mathrm{i}\frac{2\pi}{T}l\phi_n(u) } dp(u) \Big|^2] \leq C,
	\]
	under Assumption (H1).
	On the other hand,	
	\begin{align}
		\E \Big[ \Big| &\frac{1}{T}\int_{0}^{T} \int_{0}^{t} \mathrm{e}^{-\mathrm{i}\frac{2\pi}{T}l\phi_n(u)} D_M(\phi_n(t)-\phi_n(u)) \, dp(u) \,dp(t) \nonumber\\
		&+ \frac{1}{T}\int_{0}^{T} \mathrm{e}^{-\mathrm{i}\frac{2\pi}{T} l\phi_n(t)} \int_{0}^{t}  D_M(\phi_n(t)-\phi_n(u)) \, dp(u) \,dp(t) \Big|^2 \Big] \nonumber
	\end{align}
	\begin{equation}
		\label{term1}
		\leq C \E \Big[ \Big| \frac{1}{T}\int_{0}^{T} \int_{0}^{t} \mathrm{e}^{-\mathrm{i}\frac{2\pi}{T} l\phi_n(u)} D_M(\phi_n(t)-\phi_n(u)) \, dp(u) \,dp(t) \Big|^2  \Big]
	\end{equation}
	\begin{equation}
		\label{term2}
		+ C \E \Big[ \Big| \frac{1}{T}\int_{0}^{T} \mathrm{e}^{-\mathrm{i}\frac{2\pi}{T} l\phi_n(t)} \int_{0}^{t}  D_M(\phi_n(t)-\phi_n(u)) \, dp(u) \,dp(t) \Big|^2 \Big].
	\end{equation}
	
	The addends (\ref{term1}) and (\ref{term2}) have the same order of magnitude in $L_2$-norm. We then show only the estimation of the term (\ref{term1}). The latter is less than or equal to
	
	\[
	C \E \Big[ \,\, \underset{(T_1)}{\underbrace{\Big| \frac{1}{T}\int_{0}^{T} \int_{0}^{t} (\mathrm{e}^{-\mathrm{i}\frac{2\pi}{T} l\phi_n(u)} - \mathrm{e}^{-\mathrm{i}\frac{2\pi}{T} lu})  \frac{1}{2M+1} \sum_{|s|\leq M} \mathrm{e}^{-\mathrm{i}\frac{2\pi}{T} s(\phi_n(t)-\phi_n(u))}  \, dp(u) \,dp(t) \Big|^2}} \,\,\Big]
	\]
	\[
	+ C \E \Big[ \,\, \underset{(T_2)}{ \underbrace{\Big| \frac{1}{T}\int_{0}^{T} \int_{0}^{t} \mathrm{e}^{-\mathrm{i}\frac{2\pi}{T} l\phi_n(u)}  \frac{1}{2M+1} \sum_{|s|\leq M} (\mathrm{e}^{-\mathrm{i}\frac{2\pi}{T} s(\phi_n(t)-\phi_n(u))} - \mathrm{e}^{-\mathrm{i}\frac{2\pi}{T} s(t-u)})  \, dp(u) \,dp(t) \Big|^2}} \,\, \Big]
	\]
	\[
	+ C \E \Big[\,\, \underset{(T_3)}{\underbrace{\Big|\frac{1}{T} \int_0^T \int_0^t \mathrm{e}^{-\mathrm{i}\frac{2\pi}{T} lu} D_M(t-u) \, dp(u) \, dp(t)     \Big|^2}} \,\, \Big].
	\]
	The term $(T_1)$, after applying the It\^o isometry, is less than or equal to
	\[
	\underset{(T_{11})}{C \E \Big[ \int_{0}^{T} \Big| \int_{0}^{t} (\mathrm{e}^{-\mathrm{i}\frac{2\pi}{T} l\phi_n(u)} - \mathrm{e}^{-\mathrm{i}\frac{2\pi}{T} lu})  \frac{1}{2M+1} \sum_{|s|\leq M} \mathrm{e}^{-\mathrm{i}\frac{2\pi}{T} s(\phi_n(t)-\phi_n(u))}  \, dp(u) \Big|^2 \, \sigma^2(t) \, dt \Big]}
	\]
	\[
	+C \E \Big[ \int_{[0,T]^2} \Big(\int_{0}^{t} (\mathrm{e}^{-\mathrm{i}\frac{2\pi}{T} l\phi_n(u)} - \mathrm{e}^{-\mathrm{i}\frac{2\pi}{T} lu})  \frac{1}{2M+1} \sum_{|s|\leq M} \mathrm{e}^{-\mathrm{i}\frac{2\pi}{T} s(\phi_n(t)-\phi_n(u))}  \, dp(u) \Big)
	\]
	\[
	\underset{(T_{12})}{\Big(\int_{0}^{z} (\mathrm{e}^{\mathrm{i}\frac{2\pi}{T} l\phi_n(v)} - \mathrm{e}^{\mathrm{i}\frac{2\pi}{T} lv})  \frac{1}{2M+1} \sum_{|s|\leq M} \mathrm{e}^{\mathrm{i}\frac{2\pi}{T} s(\phi_n(z)-\phi_n(v))}  \, dp(v) \Big) \, a(z)\,  a(t)\, dz \, dt \Big]}
	\]
	\[
	\leq C \E \Big[ \int_0^T \int_0^t  (|l|\frac{2\pi}{T}|\phi_n(t)-t| + l^2\frac{4\pi^2}{T^2}o(|\phi_n(t)-t|^2))^2  \, du \, dt\Big]
	\]
	\[
	+C \E \Big[ \int_0^T \int_{[0,t]^2} (|l|\frac{2\pi}{T}|\phi_n(t)-t| + l^2 \frac{4\pi^2}{T^2}o(|\phi_n(t)-t|^2))
	\]
	\[
	\times (|l|\frac{2\pi}{T}|\phi_n(s)-s| + l^2\frac{4\pi^2}{T^2} o(|\phi_n(s)-s|^2))\, dv \, ds \, dt  \Big]
	\]
	\[
	+C (N^2\tau(n)+ o(1))  \E \Big[ \int_{[0,T]^2} \Big(\int_{0}^{t}  \frac{1}{2M+1} \sum_{|s|\leq M} \mathrm{e}^{-\mathrm{i}\frac{2\pi}{T} s(\phi_n(t)-\phi_n(u))} \, dp(u)
	\]
	\[
	\int_{0}^{z} \frac{1}{2M+1} \sum_{|s|\leq M} \mathrm{e}^{\mathrm{i}\frac{2\pi}{T} s(\phi_n(z)-\phi_n(v))}  \, dp(v) \Big) \, a(z)\,  a(t)\, dz \, dt \Big].
	\]
	To obtain the last inequality, we apply several times Taylor's formula, Assumption (H1) and the H\"older and Cauchy-Schwarz inequalities. Note that the first two addends from the right hand side correspond to the estimation of the term $(T_{11})$ whereas the third addend estimates from above $(T_{12})$. Thus, 
	\[
	\E[(T_1)] \leq C N^2 \tau^2(n) +o(1).
	\]
	Analogously, we can show that
	\[
	\E[(T_2)] \leq C M^2 \tau^2(n) + o(1).
	\]
	It remains to analyze the last addend of (\ref{term1}). 
	\[
	\E[(T_3)]\leq C \E \Big[ \int_0^T \Big|\int_0^t \mathrm{e}^{-\mathrm{i}\frac{2\pi}{T} lu} D_M(t-u) \, dp(u)  \Big|^2 \, \sigma^2(t) \, dt \Big]
	\]
	\[
	+  \Big[ \int_{[0,T]^2} \Big( \int_0^t \mathrm{e}^{-\mathrm{i}\frac{2\pi}{T} lu} D_M(t-u) \, dp(u) \int_0^z \mathrm{e}^{\mathrm{i}\frac{2\pi}{T} lz} D_M(z-v) \, dp(v) \Big)\, a(t) \, a(z)\, dt \,dz  \Big].
	\]
	Using iteratively the H\"older and the Cauchy-Schwarz inequalities, we obtain the term above is less than or equal to
	\[
	\leq C \E \Big[ \int_0^T \int_0^t D_M^2(t-u) \, \sigma^2(u) \, du \, dt \Big]
	+C \E \Big[\int_0^T \Big(\int_0^t |D_M(t-u)|^{p^{\prime}} \, a(u) \, du \Big)^{\frac{2}{p^{\prime}}} \, dt \Big]
	\]
	\[
	+ C \E\Big[ \int_{[0,T]^2} \Big( \int_0^t D_M^2(t-u)\, du \Big) dt dz \Big]^{\frac{1}{2}} \E\Big[\int_{[0,T]^2} \Big( \int_0^z D_M^2(z-v)\, dv \Big) dt dz \Big]^{\frac{1}{2}}
	\]
	\[
	+C \E\Big[ \int_{[0,T]^2} \Big( \int_0^t |D_M(t-u)|^{p^{\prime}}\, du \Big)^{\frac{2}{p^{\prime}}} dt dz \Big]^{\frac{1}{2}} \E\Big[\int_{[0,T]^2} \Big( \int_0^z D_M^2(z-v)\, dv \Big) dt dz \Big]^{\frac{1}{2}}
	\]
	\[
	+C \E\Big[ \int_{[0,T]^2} \Big( \int_0^t D_M^2(t-u)\, du \Big) dt dz \Big]^{\frac{1}{2}} \E\Big[\int_{[0,T]^2} \Big( \int_0^z |D_M(z-v)|^{p^{\prime}}\, dv \Big)^{\frac{2}{p^{\prime}}} dt dz \Big]^{\frac{1}{2}}
	\]
	\[
	+C \E\Big[ \int_{[0,T]^2} \Big( \int_0^t |D_M(t-u)|^{p^{\prime}}\, du \Big)^{\frac{2}{p^{\prime}}} dt dz \Big]^{\frac{1}{2}} \E\Big[\int_{[0,T]^2} \Big( \int_0^z |D_M(z-v)|^{p^{\prime}}\, dv \Big)^{\frac{2}{p^{\prime}}} dt dz \Big]^{\frac{1}{2}},
	\]
	for $p^{\prime} \in (1,2)$.
	By Lemma \ref{cor_dir}, 
	\[
	\E[(T_3)]\leq \frac{C}{2M+1}+\frac{C}{(2M+1)^{\frac{2}{p^{\prime}}}}+ \frac{C}{(2M+1)^{\frac{2+p^{\prime}}{2p^{\prime}}}}.
	\]
	
	Thus
	\[
	E[|\eta_{n,M,N}^2|]\leq C \frac{N}{\sqrt{M+1}} + C N^2\tau(n) + C NM\tau(n) +o(1),
	\]
	which converges to zero under Assumption (\ref{Star2}).
	We now show that the term (\ref{p1}) converges to zero in $L_1$-norm.
	
	\[
	\E[| \eta_{n,N}^1 - \eta|]
	\]
	\[
	=\E \Big [ \Big |\frac{1}{2 N+1}\sum_{|l|\leq N} \mathrm{i}l \frac{2\pi}{T} \sum_{i=0}^{n-1} \sum_{j=0}^{n-1}  \mathrm{e}^{\mathrm{i}\frac{2\pi}{T}l(t_i-t_j)}  \int_{t_j}^{t_{j+1}} \sigma^2(t) \, dt \int_{t_i}^{t_{i+1}} dp(u) -\int_0^t \, \eta(t) \, dt \Big| \Big]
	\]
	\begin{align}
		=& \E \Big[ \Big | \frac{1}{2 N+1}\sum_{|l|\leq N} \mathrm{i}l \frac{2\pi}{T} \int_0^T \int_0^T (\mathrm{e}^{\mathrm{i}\frac{2\pi}{T}l(\phi_n(t)-\phi_n(u))}-\mathrm{e}^{\mathrm{i}\frac{2\pi}{T}l(t-u)})  \sigma^2(u) \, du \, a(t) dt \Big | \Big ] \label{g1}\\
		&+ \E \Big[ \Big | \frac{1}{2 N+1}\sum_{|l|\leq N} \mathrm{i}l \frac{2\pi}{T} \int_0^T \int_0^T (\mathrm{e}^{\mathrm{i}\frac{2\pi}{T}l(\phi_n(t)-\phi_n(u))}-\mathrm{e}^{\mathrm{i}\frac{2\pi}{T}l(t-u)})  \sigma^2(u) \, du \, \sigma(t) \, dW(t) \Big | \Big ] \label{g2}\\
		&+ \E \Big [ \Big |\frac{1}{2 N+1}\sum_{|l|\leq N} \mathrm{i}l \frac{2\pi}{T} \int_0^T  \mathrm{e}^{-\mathrm{i}\frac{2\pi}{T}lu}   \sigma^2(u) \, du  \int_0^T \mathrm{e}^{\mathrm{i}\frac{2\pi}{T}lt}   dp(t) - \int_0^T \eta(t)\, dt  \Big |\Big ]. \label{trunc}
	\end{align}
	
	By using Taylor's formula, the  term (\ref{g1}) is less than or equal to
	\begin{align}
		C \frac{1}{2N+1} \sum_{|l|\leq N} |l| \frac{2\pi}{T} \E \Big[  \int_0^T \int_0^T \mathrm{e}^{-\mathrm{i}\frac{2\pi}{T}l(t-u)}\Big(& \frac{2\pi}{T} |l| |\phi_n(t)-t-\phi_n(u)+u| \nonumber\\
		&+ l^2 \frac{4 \pi^2}{T^2} o(|\phi_n(t)-t-\phi_n(u)+u|^2 )  du  dt \Big] \nonumber
	\end{align}
	\[
	\leq C N^2 \tau(n) +o(1).
	\]
	Analogously, term (\ref{g2}) is less than or equal to $C N^2 \tau(n)$.
	Let us analyze the term (\ref{trunc}). By using formula (\ref{mod3})
	\[
	\E \Big [ \Big |\frac{1}{2 N+1}\sum_{|l|\leq N} \mathrm{i}l \frac{2\pi}{T} \int_0^T  \mathrm{e}^{-\mathrm{i}\frac{2\pi}{T}lu}   \sigma^2(u) \, du  \int_0^T \mathrm{e}^{\mathrm{i}\frac{2\pi}{T}lt}   dp(t) - \int_0^T \eta(t)\, dt  \Big |\Big ]
	\]
	\[
	=\E \Big[ \Big | \frac{T^2}{2 N+1}\sum_{|l|\leq N} \mathrm{i}l \frac{2\pi}{T} c (l;\sigma^2) c(-l;dp) -\int_0^T \eta(t) \, dt \Big | \Big ]
	\]
	\[
	=\E \Big [ \Big | \frac{1}{2 N+1}\sum_{|l|\leq N} \Big(c(l; d \sigma^2) -\frac{1}{T} \int_0^T d \sigma^2(u)\Big)c(-l;dp)-\int_0^T \eta(t) \, dt \Big | \Big ].
	\]
	We now use the product rule and obtain
	\[
	\E \Big[\Big| \underset{M_{1,N}(T)}{\int_{0}^{T} \int_{0}^{t}  D_N(t-u) dp(u)d\sigma^2(t)}+ \underset{M_{2,N}(T)}{\int_{0}^{T} \int_{0}^{t}  D_N(t-u)\, d\sigma^2(u)dp(t)}
	\]
	\[
	\underset{M_{3,N}(T)}{-\int_{0}^{T} \int_{0}^{t} D_N(u) dp(u) d\sigma^2(t)} \underset{M_{4,N}(T)}{- \int_{0}^{T} \int_{0}^{t} D_N(t) d\sigma^2(u) dp(t)}
	\underset{M_{5,N}(T)}{- \int_{0}^{T} D_N(u) \eta(u) du}  \Big| \Big].
	\]
	Let us analyze the first double integral $M_{1,N}(T)$
	
	\[
	\mathbb{E}\Big[\Big|\int_{0}^{T} \int_{0}^{t}  D_N(t-u) dp(u)d\sigma^2(s) \Big|\Big] = \mathbb{E}\Big[\Big|\int_{0}^{T} \int_{0}^{t}  D_N(t-u) \sigma(u) dW(u) \gamma(s)dZ(s)
	\]
	\[
	+ \int_{0}^{T} \int_{0}^{t}  D_N(t-u) \sigma(u) dW(u) b(s) ds + \int_{0}^{T} \int_{0}^{t}  D_N(t-u) a(u) d(u) \gamma(s) dZ(s)
	\]
	\[
	+ \int_{0}^{T} \int_{0}^{t}  D_N(t-u) a(u) du b(s) ds \Big| \Big]
	\]
	
	The first two summands of the decomposition above have a $L_1$-norm of order $O(N^{-\frac{1}{2}})$ and the third and the fourth ones are of order $O(N^{-\frac{1}{p}})$, where $p \in (1,2)$. These estimations are performed by means of the use of Lemma \ref{cor_dir}, the H\"older and Cauchy-Schwarz inequalities. Analogous calculations follow for the terms $M_{1,N}(T)$, $M_{2,N}(T)$, $M_{3,N}(T)$, $M_{4,N}(T)$.

	By Lemma \ref{cor_dir}, we have
	
	\[
	\mathbb{E}[|M_{5,N}(2\pi)|] \leq C \mathbb{E}\Big[\sup_{t \in [0,T]} |\eta(t)|\Big] \Big( \int_{0}^{T}  |D_N(u)|^p
	du \Big)^{\frac{1}{p}} \leq \frac{C}{N^{\frac{1}{p}}}.
	\]
	
	Choosing $p \in (1,2)$ we obtain that the term $M_{5,N}(2\pi)$ converges to zero in $L_1$-norm as $N \to \infty$.
	Thus,
	
	\[
	\mathbb{E}\Big[ \Big|\frac{4\pi^2}{2N+1}\sum_{|l|\leq N} \mathrm{i}lc(l;\nu)c(-l;dp)
	-\int_{0}^{2\pi} \eta(t) dt \Big| \Big]
	\]
	\[
	\leq  \frac{C}{\sqrt{N}}+ \frac{C}{N^{\frac{1}{p}}}
	\]
	
	Therefore, under Assumption (\ref{Star2}), it follows the consistency of the estimator $\eta_{n,M,N}$.
	
\end{proof}

\begin{Remark}
	\label{zero_drift}
	In the proof of Theorem \ref{T0}, we show that the drift components of the logarithmic price and the volatility process appear in terms that are negligible in probability. To shorten the proofs of the theorems below, from now on, we omit calculations involving the drift terms $a(t)$ and $b(t)$ as they follow the same rationale of the one presented in the proof above.
\end{Remark}

In the proof of Theorem \ref{bias} and \ref{MSE}, an explicit formulation of the term $\eta^{\epsilon}_{n,M,N}$  is pivotal. Hence, 

\begin{align}
	\eta^{\epsilon}_{n,M,N}=  & \sum_{i, j, k:  i\neq j \neq k} D_M(t_{i}-t_{j}) D_N^{\prime}(t_{k}-t_{j})
	(\delta_i\delta_j\epsilon_k+ \delta_j \delta_k\epsilon_i+ \delta_k \delta_i\epsilon_j +\delta_i\epsilon_j\epsilon_k\nonumber\\ &+\delta_j\epsilon_i\epsilon_k+\delta_k\epsilon_i\epsilon_j+\epsilon_i\epsilon_j\epsilon_k) \nonumber\\
	&+  \sum_{i, j: i \neq j} D_M(t_{i}-t_{j}) D_N^{\prime}(t_{i}-t_{j}) \,
	(\delta_i \epsilon_j + \epsilon^2_i \delta_j+ \epsilon^2_i \epsilon_j+2 \delta_i\delta_j\epsilon_i + 2 \delta_i \epsilon_j \epsilon_i) \label{eta_noise}\\
	&+ \sum_{i, j}  D_N^{\prime}(t_{i}-t_{j}) \, ( \delta_j \epsilon_i + \epsilon^2_j \delta_i+ \epsilon^2_j \epsilon_i+2 \delta_j\delta_i\epsilon_j + 2 \delta_j \epsilon_i \epsilon_j). \nonumber
\end{align}

\begin{proof}[Proof of Theorem \ref{bias}]
	
	We first analyse the Bias due to the noise components.
	
	Because of Assumption (H3), it holds

	$$\begin{array}{ll}
		\E[\epsilon_i \epsilon_j\epsilon_j]&= \, 0 \, \textrm{if $i\neq j\neq k$}\\
		\E[\epsilon_i^2\epsilon_j] &= \, \left\{\begin{array}{ll}
			$0$ & \textrm{if $|i-j|\neq 1$},\\
			-\E[\zeta^3] & \textrm{if $j=i+1$},\\
			\E[\zeta^3] & \textrm{if $j=i-1$},
		\end{array} \right.
	\end{array}$$
	and
	\[
	\E[\eta^{\epsilon}_{n,M,N}]= \sum_{i=0}^{n-1} \sum_{j=0}^{n-1} D_N^{\prime}(t_{i}-t_{j}) \E[\epsilon_j^2\epsilon_i]
	+\sum_{i=0}^{n-1} \sum_{j=0}^{n-1} D_M(t_i-t_{j}) D_N^{\prime}(t_{i}-t_{j}) \E[\epsilon_i^2\epsilon_j]
	\]
	\[
	=\sum_{i=0}^{n-2}D_N^{\prime}(t_{i+1}-t_{i}) \E[\epsilon_i^2\epsilon_{i+1}]+ \sum_{i=1}^{n-1}D_N^{\prime}(t_{i-1}-t_{i}) \E[\epsilon_i^2\epsilon_{i-1}]
	\]
	\[
	+ \sum_{i=0}^{n-2}  D_M(t_{i+1}-t_{i}) D_N^{\prime}(t_{i+1}-t_{i}) \E[\epsilon_{i+1}^2\epsilon_i]+
	\sum_{i=1}^{n-1}  D_M(t_{i-1}-t_{i}) D_N^{\prime}(t_{i-1}-t_{i}) \E[\epsilon_{i-1}^2\epsilon_i]
	\]
	\[
	=(n-1) D_N^{\prime}(\frac{T}{n}) \E[\epsilon_i^2\epsilon_{i+1}]+(n-1) D_N^{\prime}(-\frac{T}{n}) \E[\epsilon_i^2\epsilon_{i-1}]
	\]
	\[
	+(n-1) D_M(\frac{T}{n}) D_N^{\prime}(\frac{T}{n}) \E[\epsilon_{i+1}^2\epsilon_i] +D_M(\frac{T}{n}) D_N^{\prime}(-\frac{T}{n}) \E[\epsilon_{i-1}^2\epsilon_i]
	\]
	\[
	=2(n-1)  D_N^{\prime}(\frac{T}{n})(D_M(\frac{T}{n})-1) \E[\zeta^3].
	\]
	
	By using Taylor's formula, it follows that $D_N^{\prime}(\frac{T}{n})\sim O\big(\frac{N^2}{n})$ and $D_M(\frac{T}{n})\sim 1- O\big(\frac{M^2}{n^2})$. Thus, under the Assumption (\ref{hyp1}), $|\E[\eta^{\epsilon}_{n,M,N}] |$ converges to zero as $N,M, n \to \infty$.

	The expected value of the term (\ref{dec1}) is equal to
	\[
	\sum_{i,j,k : \, i \neq j \neq k} D_M(t_i-t_{j}) D_N^{\prime}(t_{k}-t_{j}) \, \E[\delta_i \delta_j \delta_k]=0,
	\]
	because $\E[\delta_i \delta_j \delta_k ]=0$.
	
	The expected value of the term involving the components (\ref{dec2}) equals
	\[
	\E\Big[ \frac{1}{2N+1} \sum_{|l|\leq N} \mathrm{i}l \frac{2\pi}{T} \frac{1}{2M+1} \sum_{|s|\leq M}
	\sum_{i=0}^{n-1} \sum_{j=0}^{n-1} \mathrm{e}^{\mathrm{i}\frac{2\pi}{T}(l-s)(t_i-t_j)} \, \delta_i^2 \delta_j
	\]
	\[
	+\frac{1}{2 N+1}\sum_{|l|\leq N} \mathrm{i}l \frac{2\pi}{T} \sum_{i=0}^{n-1} \sum_{j=0}^{n-1}  \mathrm{e}^{\mathrm{i}\frac{2\pi}{T}l(t_i-t_j)} \, \delta_i \delta_j^2 -\int_0^T \eta(t) dt\Big]
	\]
	\[
	= \underset{(A_1)}{\E\Big[ \frac{1}{2N+1} \sum_{|l|\leq N} \mathrm{i}l \frac{2\pi}{T} \frac{1}{2M+1} \sum_{|s|\leq M}
		\sum_{i=0}^{n-1} \sum_{j=0}^{n-1} \mathrm{e}^{\mathrm{i}\frac{2\pi}{T}(l-s)(t_i-t_j)} \int_{t_i}^{t_{i+1}} \sigma^2(u) \, du \int_{t_j}^{t_{j+1}} dp(t)\Big]}
	\]
	\[
	+ \underset{(A_2)}{\E\Big[\frac{1}{2N+1} \sum_{|l|\leq N} \mathrm{i}l \frac{2\pi}{T} \frac{1}{2M+1} \sum_{|s|\leq M}
		\sum_{i=0}^{n-1} \sum_{j=0}^{n-1} \mathrm{e}^{\mathrm{i}\frac{2\pi}{T}(l-s)(t_i-t_j)} \int_{t_i}^{t_{i+1}} \int_{t_i}^{t} dp(u) dp(t) \int_{t_j}^{t_{j+1}} dp(t) \Big]}
	\]
	\[
	+\underset{(A_3)}{\E\Big[\frac{1}{2 N+1}\sum_{|l|\leq N} \mathrm{i}l \frac{2\pi}{T} \sum_{i=0}^{n-1} \sum_{j=0}^{n-1}  \mathrm{e}^{\mathrm{i}\frac{2\pi}{T}l(t_i-t_j)}  \int_{t_j}^{t_{j+1}}  \int_{t_j}^{t} dp(u) dp(t) \int_{t_i}^{t_{i+1}} dp(t)\Big]}
	\]
	\[
	+\underset{(A_4)}{\E\Big[ \frac{1}{2 N+1}\sum_{|l|\leq N} \mathrm{i}l \frac{2\pi}{T} \sum_{i=0}^{n-1} \sum_{j=0}^{n-1}  \mathrm{e}^{\mathrm{i}\frac{2\pi}{T}l(t_i-t_j)}  \int_{t_j}^{t_{j+1}} \sigma^2(u) \, du \int_{t_i}^{t_{i+1}} dp(t)   -\int_0^T \eta(t) dt \Big]}.
	\]

	The term $(A_1)$ can be further decomposed in
	\[
	\underset{(A_1^*)}{\E\Big[ \frac{1}{2N+1} \sum_{|l|\leq N} \mathrm{i}l \frac{2\pi}{T} \frac{1}{2M+1} \sum_{|s|\leq M}
		\sum_{i=1}^{n-1} \sum_{j=0}^{i-1} \mathrm{e}^{\mathrm{i}\frac{2\pi}{T}(l-s)(t_i-t_j)} \int_{t_i}^{t_{i+1}} \sigma^2(u) \, du \int_{t_j}^{t_{j+1}} dp(t)\Big]}
	\]
	\[
	+\E\Big[ \frac{1}{2N+1} \sum_{|l|\leq N} \mathrm{i}l \frac{2\pi}{T} \frac{1}{2M+1} \sum_{|s|\leq M}
	\sum_{j=1}^{n-1} \sum_{i=0}^{j-1} \mathrm{e}^{\mathrm{i}\frac{2\pi}{T}(l-s)(t_i-t_j)} \int_{t_i}^{t_{i+1}} \sigma^2(u) \, du \int_{t_j}^{t_{j+1}} dp(t)\Big].
	\]
	By applying the tower property with respect to the sigma-algebra $\mathcal{F}_{i+1}$, the second summand is zero because of the martingale property of the It\^o integrals.
	Thus, the term $(A_1)$ is just equal to the term $(A_1^*)$.
	We call $|(A_1^*)|=\Gamma(n,M,N)$.

	\[
	\Gamma(n,M,N)= \Big | \E \Big[ \frac{1}{2N+1} \sum_{|l|\leq N} \mathrm{i}l \frac{2\pi}{T} \int_{0}^{T} \int_0^t \mathrm{e}^{\mathrm{i}\frac{2\pi}{T}l(\phi_n(t)-\phi_n(u))} \frac{1}{2M+1}
	\]
	\[
	\sum_{|s|\leq M}\mathrm{e}^{-\mathrm{i}\frac{2\pi}{T}s(\phi_n(t)-\phi_n(u))} \,  dp(u) \, \sigma^2(t) \, dt  \Big] \Big|
	\]
	\[
	= \Big| \E\Big[ \frac{1}{2N+1} \sum_{|l|\leq N} \mathrm{i}l \frac{2\pi}{T} \int_{0}^{T} \int_0^t \mathrm{e}^{\mathrm{i}\frac{2\pi}{T}l(\phi_n(t)-\phi_n(u))} \frac{1}{2M+1}
	\]
	\[
	\sum_{|s|\leq M} (\mathrm{e}^{-\mathrm{i}\frac{2\pi}{T}s(\phi_n(t)-\phi_n(u))}- \mathrm{e}^{-\mathrm{i}\frac{2\pi}{T}s(t-u)}) \,  dp(u) \, \sigma^2(t) \, dt  \Big] \Big|
	\]
	\[
	+ \Big | \E\Big[ \frac{1}{2N+1} \sum_{|l|\leq N} \mathrm{i}l \frac{2\pi}{T} \int_{0}^{T} \int_0^t (\mathrm{e}^{-\mathrm{i}\frac{2\pi}{T}l(\phi_n(t)-\phi_n(u))} - \mathrm{e}^{-\mathrm{i}\frac{2\pi}{T}l(t-u)}) \]
	\[
	\frac{1}{2M+1} \sum_{|s|\leq M} \mathrm{e}^{-\mathrm{i}\frac{2\pi}{T}s(t-u)} \,  dp(u) \, \sigma^2(t) \, dt \Big] \Big|
	\]
	\[
	+  \Big | \E\Big[ \frac{1}{2N+1} \sum_{|l|\leq N} \mathrm{i}l \frac{2\pi}{T} \int_{0}^{T} \int_0^t \mathrm{e}^{-\mathrm{i}\frac{2\pi}{T}l(t-u)} D_M(t-u) \,  dp(u) \, \sigma^2(t) \, dt  \Big] \Big|.
	\]

	The third summand is less than or equal to
	\[
	\E \Big[\frac{1}{2N+1} \sum_{|l|\leq N} |l| \frac{2\pi}{T} \Big| \int_{0}^{T} \int_0^t \mathrm{e}^{-\mathrm{i}\frac{2\pi}{T}(l)(t-u)} D_M(t-u) \,  dp(u) \, \sigma^2(t) \, dt \Big|  \Big]
	\]
	\[
	\leq \frac{1}{2N+1} \sum_{|l|\leq N} |l| \frac{2\pi}{T} \, \E\Big[\sup_{[0,T]}\sigma^2(t)\Big]^{\frac{3}{2}} \, T \, \E\Big[ \int_{0}^T D_M^2(u) du \Big]^{\frac{1}{2}}
	\]
	\[ \leq 2\pi N \E\Big[\sup_{[0,T]}\sigma^2(t)\Big]^{\frac{3}{2}} \Big(\frac{T}{2M+1}\Big)^{\frac{1}{2}}
	\]
	
	by using the Cauchy Schwartz and H\"older inequality, the It\^o isometry and the properties of the rescaled Dirichlet kernel.
	
	By means of the Taylor's formula, we  obtain estimations for the first and second summand of $\Gamma(n,M,N)$ as follows

	\[
	\Big| \E\Big[ \frac{1}{2N+1} \sum_{|l|\leq N} \mathrm{i}l \frac{2\pi}{T} \int_{0}^{T} \int_0^t \mathrm{e}^{\mathrm{i}\frac{2\pi}{T}(l)(\phi_n(t)-\phi_n(u))} \frac{1}{2M+1} \]
	\[
	\sum_{|s|\leq M} (\mathrm{e}^{-\mathrm{i}\frac{2\pi}{T}(s)(\phi_n(t)-\phi_n(u))}- \mathrm{e}^{-\mathrm{i}\frac{2\pi}{T}(s)(t-u)}) \,  dp(u) \, \sigma^2(t) \, dt  \Big] \Big|
	\]
	\[
	\leq \E\Big[ \Big| \frac{1}{2N+1} \sum_{|l|\leq N} \mathrm{i}l \frac{2\pi}{T} \int_{0}^{T} \int_0^t \mathrm{e}^{\mathrm{i}\frac{2\pi}{T}(l)(\phi_n(t)-\phi_n(u))} \frac{1}{2M+1}
	\]
	\[
	\sum_{|s|\leq M}  \mathrm{e}^{-\mathrm{i}\frac{2\pi}{T}(s)(t-u)}(s \frac{2\pi}{T}(t-u-\phi_n(t)+\phi_n(u))+o(1)) \,  dp(u) \, \sigma^2(t) \, dt \Big| \Big]
	\]	
	\[
	\leq \frac{1}{2N+1} \sum_{|l|\leq N} |l| \frac{2\pi}{T} \frac{1}{2M+1}  \sum_{|s|\leq M} \E\Big[\sup_{[0,T]}\sigma^2(t)\Big]
	\]	
	\[
	\int_0^T \E \Big[ \Big| \int_0^t \mathrm{e}^{\mathrm{i}\frac{2\pi}{T}(l)(\phi_n(t)-\phi_n(u))}  \mathrm{e}^{-\mathrm{i}\frac{2\pi}{T}(s)(t-u)} (s\frac{2\pi}{T}(t-u-\phi_n(t)+\phi_n(u))+o(1)) \,\, dp(u) \Big| \Big] dt
	\]
	\[
	\leq \frac{1}{2N+1} \sum_{|l|\leq N} |l| \frac{2\pi}{T} \frac{1}{2M+1}  \sum_{|s|\leq M} \E\Big[\sup_{[0,T]}\sigma^2(t)\Big]
	\]
	\[
	\int_0^T \E \Big[\int_0^t (s^2 \frac{4\pi^2}{T^2}(t-u-\phi_n(t)+\phi_n(u))^2+o(1)) \,\, \sigma^2(u) du \Big]^{\frac{1}{2}} dt
	\]
	\[
	\leq \frac{1}{2N+1} \sum_{|l|\leq N} |l| \frac{2\pi}{T} \frac{1}{2M+1}  \sum_{|s|\leq M} \E\Big[\sup_{[0,T]}\sigma^2(t)\Big]^{\frac{3}{2}}\int_0^T \Big(\int_0^t s^2 4\frac{4\pi^2}{n^2} +o(1) \, du \Big)^{\frac{1}{2}} dt
	\]
	\[
	\leq \frac{MN}{n} \, 8 \pi^2 T^{\frac{1}{2}}\, \E\Big[\sup_{[0,T]}\sigma^2(t)\Big]^{\frac{3}{2}} +o(1),
	\]
	and
	\[
	\Big | \E\Big[ \frac{1}{2N+1} \sum_{|l|\leq N} \mathrm{i}l \frac{2\pi}{T} \int_{0}^{T} \int_0^t (\mathrm{e}^{-\mathrm{i}\frac{2\pi}{T}(l)(\phi_n(t)-\phi_n(u))} - \mathrm{e}^{-\mathrm{i}\frac{2\pi}{T}(l)(t-u)}) \]
	\[
	\frac{1}{2M+1} \sum_{|s|\leq M} \mathrm{e}^{-\mathrm{i}\frac{2\pi}{T}s(t-u)} \,  dp(u) \, \sigma^2(t) \, dt \Big] \Big|
	\]
	\[
	\leq \frac{1}{2N+1} \sum_{|l|\leq N} |l| \frac{2\pi}{T} \frac{1}{2M+1}  \sum_{|s|\leq M} \E\Big[\sup_{[0,T]}\sigma^2(t)\Big]^{\frac{3}{2}}
	\]
	\[
	\times \int_0^T \Big(\int_0^t  (l^2 \frac{4\pi^2}{T^2} (\phi_n(t)-\phi_n(u)-t+u)^2+o(1)) \,\, du \Big)^{\frac{1}{2}} dt
	\]
	\[
	\leq \frac{N^2}{n} \, 8 \pi^2 T^{\frac{1}{2}} \, \E\Big[\sup_{[0,T]}\sigma^2(t)\Big]^{\frac{3}{2}} +o(1).
	\]

	Let us now further decompose the term $(A_4)$ in
	\[
	(A_4)=  \underset{(A_{4.1})}{\E \Big[ \frac{1}{2 N+1}\sum_{|l|\leq N} \mathrm{i}l \frac{2\pi}{T} \int_0^T \mathrm{e}^{-\mathrm{i}l \frac{2\pi}{T} \phi_n(u)} \sigma^2(u) \, du \, \Big( \int_0^T \mathrm{e}^{\mathrm{i}l \frac{2\pi}{T} \phi_n(t)} - \mathrm{e}^{\mathrm{i}l \frac{2\pi}{T} t} dp(t) \Big)\Big]}
	\]
	\[
	+ \underset{(A_{4.2})}{ \E \Big[ \frac{1}{2 N+1}\sum_{|l|\leq N} \mathrm{i}l \frac{2\pi}{T} \int_0^T \mathrm{e}^{\mathrm{i}l \frac{2\pi}{T} \phi_n(t)} \, dp(t) \, \Big( \int_0^T \mathrm{e}^{-\mathrm{i}l \frac{2\pi}{T} \phi_n(u)} - \mathrm{e}^{-\mathrm{i}l \frac{2\pi}{T} u} \sigma^2(u) \, du \Big)\Big]}
	\]
	\[
	+ \underset{(A_{4.3})}{ \E \Big[ \frac{1}{2 N+1}\sum_{|l|\leq N} \mathrm{i}l \frac{2\pi}{T} \int_0^T \mathrm{e}^{-\mathrm{i}l \frac{2\pi}{T} u} \, \sigma^2(u) \, du \,  \int_0^T \mathrm{e}^{\mathrm{i}l \frac{2\pi}{T} t} \, dp(t)  - \int_0^T \eta(t) dt \Big]} .
	\]
	We call $\Lambda(n,N)=|(A_2)|+|(A_3)|+|(A_{4.1})|+|(A_{4.2})|$.
	
	Let us first discuss the terms $(A_2)$ and $(A_3)$. For $i\neq j$, the terms
	\[
	\frac{1}{2N+1} \sum_{|l|\leq N} \mathrm{i}l \frac{2\pi}{T} \frac{1}{2M+1} \sum_{|s|\leq M}
	\sum_{i\neq j} \mathrm{e}^{\mathrm{i}\frac{2\pi}{T}(l-s)(t_i-t_j)} \,\, \E\Big[\int_{t_i}^{t_{i+1}} \int_{t_i}^{t} dp(u) dp(t) \int_{t_j}^{t_{j+1}} dp(t) \Big]
	\]
	and
	\[
	\frac{1}{2 N+1}\sum_{|l|\leq N} \mathrm{i}l \frac{2\pi}{T} \sum_{i \neq j}  \mathrm{e}^{\mathrm{i}\frac{2\pi}{T}l(t_i-t_j)} \,\,  \E \Big[ \int_{t_j}^{t_{j+1}}  \int_{t_j}^{t} dp(u) dp(t) \int_{t_i}^{t_{i+1}} dp(t)\Big]
	\]
	are zero because the It\^o integrals appearing in the expectations are defined on non overlapping intervals.
	For $i=j$, let us evaluate the terms $|(A_2)|$ and $|(A_3)|$. In this instance, $(A_2)$ and $(A_3)$ are both equal
	and
	\[
	\Big| \E \Big[ \frac{1}{2 N+1}\sum_{|l|\leq N} \mathrm{i}l \frac{2\pi}{T} \sum_{i=0}^{n-1} \int_{t_i}^{t_{i+1}} \int_{t_i}^{t} dp(u) dp(t) \int_{t_i}^{t_{i+1}} dp(t) \Big] \Big|
	\]
	\[
	\leq \Big| \frac{1}{2 N+1}\sum_{|l|\leq N} \mathrm{i}l \frac{2\pi}{T} \sum_{i=0}^{n-1} \E \Big[  \int_{t_i}^{t_{i+1}} \int_{t_i}^{t} dp(u) \sigma^2(t) dt \Big] \Big|
	\]
	\[
	\leq \frac{1}{2 N+1}\sum_{|l|\leq N} |l| \frac{2\pi}{T} \sum_{i=0}^{n-1} \E \Big[ \Big| \int_{t_i}^{t_{i+1}} \int_{t_i}^{t} dp(u) \sigma^2(t) dt \Big| \Big]
	\]
	\[
	\leq N \frac{2\pi}{T} \sum_{i=0}^{n-1} \E\Big[\sup_{[0,T]}\sigma^2(t)\Big]^{\frac{3}{2}} (t_{i+1}-t_i)^{\frac{3}{2}}
	\]
	\[
	\leq  N \frac{2\pi}{T} \Big(\frac{T}{n} \Big)^{\frac{3}{2}} n \E\Big[\sup_{[0,T]}\sigma^2(t)\Big]^{\frac{3}{2}} =
	\frac{N}{\sqrt{n}} \,  2\pi T^{\frac{1}{2}} \,  \E\Big[\sup_{[0,T]}\sigma^2(t)\Big]^{\frac{3}{2}},
	\]
	by using the It\^o isometry and the H\"older inequality.
	Moreover, because of the Cauchy-Schwartz inequality, $|(A_{4.1})|$ is less than or equal to
	
	\[
	\frac{1}{2 N+1}\sum_{|l|\leq N} |l| \frac{2\pi}{T} \E \Big[ \Big| \int_0^T \mathrm{e}^{-\mathrm{i}l \frac{2\pi}{T} \phi_n(u)} \sigma^2(u) \, du \Big|^2 \Big]^{\frac{1}{2}} \E \Big[ \Big| \int_0^T \mathrm{e}^{\mathrm{i}l \frac{2\pi}{T} \phi_n(t)} - \mathrm{e}^{\mathrm{i}l \frac{2\pi}{T} \phi_n(t)} dp(t) \Big|^2 \Big]^{\frac{1}{2}}
	\]
	\[
	\leq \frac{1}{2 N+1}\sum_{|l|\leq N} |l| \frac{2\pi}{T} \E \Big[ \int_0^T  \sigma^4(u) \, du \Big]^{\frac{1}{2}} T^{\frac{1}{2}} \E \Big[ \int_0^T (l \frac{2\pi}{T} \frac{T}{n}+o(1))^2 \sigma^2(t) \, dt  \Big]^{\frac{1}{2}}
	\]
	\[
	\frac{1}{2 N+1}\sum_{|l|\leq N} |l| \frac{2\pi}{T} \E [\sup_{t\in[0,T]} \sigma^4(t)]^{\frac{1}{2}}  \E [\sup_{t\in[0,T]} \sigma^2(t)]^{\frac{1}{2}} T \frac{N}{n} 2\pi +o(1)
	\]
	\[
	\leq \frac{N^2}{n} \, 4 \pi^2 \,  \E [\sup_{t\in[0,T]} \sigma^4(t)]^{\frac{1}{2}}  \E [\sup_{t\in[0,T]} \sigma^2(t)]^{\frac{1}{2}} +o(1),
	\]
	by applying Taylor's Formula and the Cauchy-Schwartz inequality.
	Analogously, it can be shown that $|(A_{4.2})|$ is less than or equal to
	
	$$\frac{N^2}{n} \, 4 \pi^2 T^{\frac{1}{2}} \,  \E [\sup_{t\in[0,T]} \sigma^4(t)]^{\frac{1}{2}}  \E [\sup_{t\in[0,T]} \sigma^2(t)]^{\frac{1}{2}} +o(1).$$
	
	It remains to evaluate the term $|(A_{4.3})|$ that we call $\Psi(N)$.
	By formula (\ref{mod3}), $|(A_{4.3})|$ can be expressed as
	
	\[
	\Big| \E \Big[ \int_{0}^{T} \int_{0}^{t}  D_N(t-u) dp(u)d\sigma^2(t)+ \int_{0}^{T} \int_{0}^{t}  D_N(t-u)\, d\sigma^2(u)dp(t)
	\]
	\[
	-\int_{0}^{T} \int_{0}^{t} D_N(u) dp(u) d\sigma^2(t)- \int_{0}^{T} \int_{0}^{t} D_N(t) d\sigma^2(u) dp(t)
	- \int_{0}^{T} D_N(u) \eta(u) du \Big] \Big|.
	\]
	The It\^o integrals have zero mean and the terms above are simply equal to
	\[
	\Big| \E \Big[- \int_{0}^{T} D_N(u) \eta(u) du \Big] \Big| \leq \E\Big[ \int_0^T D_N^2(u) du \int_0^T \eta(u)^2 du \Big]^{\frac{1}{2}}
	\]
	after applying the Cauchy-Schwartz inequality.
	Thus,
	\[
	\Psi(N) \leq \frac{T}{\sqrt{2N+1}} \, \E[\sup_{t\in[0,T]} \eta(t)^2]^{\frac{1}{2}}.
	\]
	
	The terms $\Gamma(n,M,N)$, $\Lambda(n,N)$ and $\Omega(N)$ converge to zero under Assumption (\ref{hyp1}) which concludes the proof. 
\end{proof}

\begin{proof}[Proof of Theorem \ref{MSE}]
	We have that
	\begin{align}
		\label{mse}
		\E[(\widetilde{\eta}_{n,M,N}-\eta)^2]	&= \E \Big[ \Big( \sum_{i, j, k:  i\neq j \neq k} D_M(t_{i}-t_{j}) D_N^{\prime}(t_{k}-t_{j})
		\widetilde{\delta}_{i} \widetilde{\delta}_{j} \widetilde{\delta}_{k}\\
		&+ \sum_{i, j: i \neq j} D_M(t_{i}-t_{j}) D_N^{\prime}(t_{i}-t_{j})
		\widetilde{\delta}_{i}^2 \widetilde{\delta}_{j} 
		+ \sum_{i, j}  D_N^{\prime}(t_{i}-t_{j}) \widetilde{\delta}_{i} \widetilde{\delta}_{j}^2 -\eta \Big)^2 \Big] \nonumber
	\end{align}
	which is in turn equal to
	\begin{align*}
		& \E \Big[ \Big( (\eta_{n,M,N} -\eta )\\
		&+ \Big(\sum_{i, j, k:  i\neq j \neq k} D_M(t_{i}-t_{j}) D_N^{\prime}(t_{k}-t_{j})
		(\delta_i\delta_j\epsilon_k+ \delta_j \delta_k\epsilon_i+ \delta_k \delta_i\epsilon_j\\
		&+\delta_i\epsilon_j\epsilon_k +\delta_j\epsilon_i\epsilon_k+\delta_k\epsilon_i\epsilon_j+\epsilon_i\epsilon_j\epsilon_k) \Big)
		\\
		&+ \Big( \sum_{i, j: i \neq j} D_M(t_{i}-t_{j}) D_N^{\prime}(t_{i}-t_{j}) \,
		(\delta_i \epsilon_j + \epsilon^2_i \delta_j+ \epsilon^2_i \epsilon_j+2 \delta_i\delta_j\epsilon_i + 2 \delta_i \epsilon_j \epsilon_i) \Big) \\
		&+ \Big( \sum_{i, j}  D_N^{\prime}(t_{i}-t_{j}) \, ( \delta_j \epsilon_i + \epsilon^2_j \delta_i+ \epsilon^2_j \epsilon_i+2 \delta_j\delta_i\epsilon_j + 2 \delta_j \epsilon_i \epsilon_j) \Big) \Big)^2 \Big].
	\end{align*}

	The term $\E[(\eta_{n,M,N}-\eta)^2]$ corresponds to the mean squared error of the estimator (\ref{decomp}) in the absence of microstructure noise and converges to zero as $n,M,N$ tend to infinity. The error decomposition (\ref{est_error}) and the proof of Theorem \ref{T0} highlight that the term $\eta_{n,N}^1$ is in $L_2$-norm bigger than $(\eta_{n,M,N}^2-\eta)$.
	Then, to prove our claim, it is enough to analyse the convergence to zero of 
	\begin{equation}
		\label{mse_1}
		\E [ (\eta_{n,N}^1 -\eta)^2].
	\end{equation}
	Following Remark \ref{zero_drift}, we have that (\ref{mse_1}) is equal to
	\begin{align}
		\E \Big[ \Big( &\frac{1}{2 N+1}\sum_{|l|\leq N} \mathrm{i}l \frac{2\pi}{T} \int_0^T \int_0^T (\mathrm{e}^{\mathrm{i}\frac{2\pi}{T}l(\phi_n(t)-\phi_n(u))}-\mathrm{e}^{\mathrm{i}\frac{2\pi}{T}l(t-u)})  \sigma^2(u) \, du \, \sigma(t) \, dW(t) \nonumber\\
		&+\int_0^T \int_0^t D_N(t-u) \, dp(u) \, d\sigma^2(t) + \int_0^T \int_0^t D_N(t-u) \, d\sigma^2(t)\, dp(u) \nonumber\\
		&-\int_{0}^{T} \int_{0}^{t} D_N(u) dp(u) d\sigma^2(t)-\int_{0}^{T} \int_{0}^{t} D_N(u)  d\sigma^2(t) dp(u)-\int_0^T D_N(u) \, \eta(u) du \Big)^2 \Big] \nonumber
	\end{align}
	\begin{align}
		\leq &2 \E \Big[ \Big( \frac{1}{2 N+1}\sum_{|l|\leq N} \mathrm{i}l \frac{2\pi}{T} \int_0^T \int_0^T (\mathrm{e}^{\mathrm{i}\frac{2\pi}{T}l(\phi_n(t)-\phi_n(u))}-\mathrm{e}^{\mathrm{i}\frac{2\pi}{T}l(t-u)})  \sigma^2(u) \, du \, \sigma(t) \, dW(t) \Big)^2 \Big] \label{e1}\\
		&+8 \E\Big[ \Big( \int_0^T \int_0^t D_N(t-u) \, dp(u) \, d\sigma^2(t) \Big)^2 \Big] +8 \E\Big[ \Big( \int_0^T \int_0^t D_N(t-u) \, dp(u) \, d\sigma^2(t) \Big)^2 \Big] \label{e2}\\
		&+8 \E \Big[ \Big( \int_0^T D_N(u) \, \eta(u) du \Big)^2 \Big] +8 \E \Big[ \Big( \int_0^T D_N(u) \, \eta(u) du \Big)^2 \Big] \label{e3}
	\end{align}
	\begin{align}
		\leq & \frac{128 \pi^2}{T^2} \frac{N^4}{n^2} \E[\sup_{t \in [0,T]} \sigma^2(t)] \E[\sup_{t \in [0,T]} \sigma^4(t)] \label{e11}\\
		&+ 16 \frac{T^2}{2N+1} \E[\sup_{t \in [0,T]} \sigma^2(t)] \E[\sup_{t \in [0,T]} \gamma^2(t)] \label{e22}\\
		&+ 16 \frac{T}{2N+1}\E [\sup_{t \in [0,T]} \eta(t)^2], \label{e33}
	\end{align}
	where (\ref{e11}), (\ref{e22}), (\ref{e33}) correspond to the estimation of the summands (\ref{e1}), (\ref{e2}), (\ref{e3}), respectively. Thus, (\ref{mse_1}) converges to zero as $n,N \to \infty$ and so does the mean squared error of the estimator (\ref{decomp}).
	However, whenever a noise component appears in the decomposition (\ref{mse}), the related terms diverge to infinity as $n,M,N$ goes to infinity.
	As exemplary calculation, we will show that
	\begin{equation}
		\label{mse_dom}
		\E[\Big( \sum_{i, j}  D_N^{\prime}(t_{i}-t_{j}) \, ( \delta_j \epsilon_i + \epsilon^2_j \delta_i+ \epsilon^2_j \epsilon_i+2 \delta_j\delta_i\epsilon_j + 2 \delta_j \epsilon_i \epsilon_j) \Big) \Big)^2 ]
	\end{equation}
	diverges as $n,N \to \infty$ and is greater than $O(n^2N)$. In order to handle the other terms in (\ref{mse}), the strategies of computation addressed below have to be used. Ultimately, this leads to show that the remaining terms in (\ref{mse}) are greater than $O(NM^2+\frac{n^2N}{M})$.

	We have that
	\begin{align}
		& \E\Big[\Big( \sum_{i, j}  D_N^{\prime}(t_{i}-t_{j}) \,  \delta_j \epsilon_i + \epsilon^2_j \delta_i+ \epsilon^2_j \epsilon_i+2 \delta_j\delta_i\epsilon_j + 2 \delta_j \epsilon_i \epsilon_j \Big)^2 \Big] \nonumber\\
		&= \sum_{i, j} (D_N^{\prime}(t_{i}-t_{j}))^2  \E[( \delta_j \epsilon_i + \epsilon^2_j \delta_i+ \epsilon^2_j \epsilon_i+2 \delta_j\delta_i\epsilon_j + 2 \delta_j \epsilon_i \epsilon_j)^2] \label{inf1}\\
		&+ \sum_{i, j,i^{\prime},j^{\prime}: i\neq i^{\prime}, j\neq j^{\prime}} D_N^{\prime}(t_{i}-t_{j}) \overline{D_N^{\prime}(t_{i^{\prime}}-t_{j^{\prime}})}
		\E[(\delta_j \epsilon_i + \epsilon^2_j \delta_i+ \epsilon^2_j \epsilon_i+2 \delta_j\delta_i\epsilon_j + 2 \delta_j \epsilon_i \epsilon_j) \nonumber\\
		& (\delta_{j^{\prime}} \epsilon_{i^{\prime}} + \epsilon^2_{j^{\prime}} \delta_{i^{\prime}}+ \epsilon^2_{j^{\prime}} \epsilon_{i^{\prime}}+2 \delta_{j^{\prime}}\delta_{i^{\prime}}\epsilon_{j^{\prime}} + 2 \delta_{j^{\prime}} \epsilon_{i^{\prime}} \epsilon_{j^{\prime}})] \label{inf2}.
	\end{align}
	Under Assumption (H2), we have that (\ref{mse_dom}) is equal to
	
	\begin{align}
		&\sum_{i, j} (D_N^{\prime}(t_{i}-t_{j}))^2  (\E[ \delta_j^2] \E[\epsilon_i^2] + \E[\epsilon^4_j] \E[\delta_i^2]+ \E[\epsilon^4_j \epsilon_i^2]+ 4 \E[\delta_j^2\delta_i^2] \E[\epsilon_j^2] + 4 \E[\delta_j^2] \E[\epsilon_i^2 \epsilon_j^2]) \label{mse_dom1}\\
		+&\sum_{i, j,i^{\prime},j^{\prime}: i\neq i^{\prime}, j\neq j^{\prime}} D_N^{\prime}(t_{i}-t_{j}) \overline{D_N^{\prime}(t_{i^{\prime}}-t_{j^{\prime}})} \E [ \epsilon^2_j \epsilon_i \epsilon^2_{j^{\prime}} \epsilon_{i^{\prime}} ] \label{mse_dom2}.
	\end{align}
	
	It holds that
	
	\begin{align}
		&\E[\epsilon_i^2]=2\E[\zeta^2] \nonumber\\
		&\E[\epsilon_i^4]=2\E[\zeta^4] +6 \E[\zeta^2]^2 \nonumber \\
		&\begin{array}{ll}
			\E[\epsilon_j^4\epsilon_i^2] &= \, \left\{\begin{array}{ll}
				4\E[\zeta^4]\E[\zeta^2]+12\E[\zeta^2]^3 & \textrm{if $|i-j|\neq 1$},\\
				9\E[\zeta^4]\E[\zeta^2]+\E[\zeta^6]+6\E[\zeta^2]^3-4\E[\zeta^3]^2 & \textrm{if $i=j-1$},\\
				13\E[\zeta^4]\E[\zeta^2]+\E[\zeta^6]+2\E[\zeta^2]^3-4\E[\zeta^3]^2 & \textrm{if $i=j+1$}.
			\end{array} \right.
		\end{array} \nonumber \\
		&\begin{array}{ll}
			\E[\epsilon_j^2\epsilon_i^2] &= \, \left\{\begin{array}{ll}
				4\E[\zeta^2]^2 & \textrm{if $|i-j|> 1$},\\
				3\E[\zeta^2]^2+\E[\zeta^4] & \textrm{if $i=j-1$},\\
				3\E[\zeta^2]^2+\E[\zeta^4] & \textrm{if $i=j+1$}.
			\end{array} \right.
		\end{array} \nonumber \\
		&\E[\epsilon_i^3]= 0 \nonumber
	\end{align}
	\begin{align}
		\begin{array}{ll}
			\E [ \epsilon^2_j \epsilon_i \epsilon^2_{j^{\prime}} \epsilon_{i^{\prime}} ] &= \, \left\{\begin{array}{ll}
				0 & \textrm{if $i\neq i^{\prime},j\neq j^{\prime}, i\neq j^{\prime}, j\neq i^{\prime}$},\\
				0 & \textrm{if $i\neq i^{\prime},j\neq j^{\prime}, i= j^{\prime}, j= i^{\prime}$ and $|i-j|\neq 1$},\\
				a & \textrm{if $i\neq i^{\prime},j\neq j^{\prime}, i= j^{\prime}, j= i^{\prime}$ and $i=j+1$}, \\
				b & \textrm{if $i\neq i^{\prime},j\neq j^{\prime}, i= j^{\prime}, j= i^{\prime}$ and $i=j-1$},
			\end{array} \right.
		\end{array} \nonumber
	\end{align}
	where $a=\E[\zeta^3]^2-\E[\zeta^6] -6\E[\zeta^4]\E[\zeta^2]-9\E[\zeta^2]^3$, and $b=\E[\zeta^3]^2-\E[\zeta^6] -6\E[\zeta^4]\E[\zeta^2]-9\E[\zeta^2]^3$.
	
	Therefore (\ref{mse_dom}) is equal to
	\begin{align}
		&\sum_{i,j} (D_N^{\prime}(t_{i}-t_{j}))^2 (2 \E[\delta_j^2] \E[\zeta^2] + 2\E[\delta_i^2]\E[\zeta^4] +6 \E[\delta_i^2]\E[\zeta^2]^2 + 8\E[\delta_i^2\delta_j^2] \E[\zeta^2] ) \label{n1}\\
		+& \sum_{i,j:|i-j|\neq 1} (D_N^{\prime}(t_{i}-t_{j}))^2 16 \E[\delta_j^2] \E[\zeta^2]^2 +\sum_{i,j:|i-j|=1}  (D_N^{\prime}(t_{i}-t_{j}))^2 \E[\delta_j^2] \nonumber\\
		&\times (12 \E[\zeta^2]^2+ 4\E[\zeta^4]) \label{n2}\\
		+& \sum_{i,j: |i-j|\neq 1} (D_N^{\prime}(t_{i}-t_{j}))^2 (4\E[\zeta^4]\E[\zeta^2]+12\E[\zeta^2]^3)\label{n3}\\
		+& \sum_{i,j: i=j-1} (D_N^{\prime}(t_{j-1}-t_{j}))^2 (9\E[\zeta^4]\E[\zeta^2]+\E[\zeta^6]+6\E[\zeta^2]^3-4\E[\zeta^3]^2) \label{s1}\\
		+& \sum_{i,j: i=j+1} (D_N^{\prime}(t_{j+1}-t_{j}))^2 (13\E[\zeta^4]\E[\zeta^2]+\E[\zeta^6]+2\E[\zeta^2]^3-4\E[\zeta^3]^2) \label{s2}\\
		+& \sum_{i,j:i=j-1} D_N^{\prime}(t_{j-1}-t_{j}) \overline{D_N^{\prime}(t_{j}-t_{j-1})} (\E[\zeta^3]^2-\E[\zeta^6] -6\E[\zeta^4]\E[\zeta^2]-9\E[\zeta^2]^3) \label{s3}\\
		+& \sum_{i,j:i=j+1} D_N^{\prime}(t_{j+1}-t_{j}) \overline{D_N^{\prime}(t_{j}-t_{j+1})} (\E[\zeta^3]^2-\E[\zeta^6] -6\E[\zeta^4]\E[\zeta^2]-9\E[\zeta^2]^3) \label{s4}
	\end{align}
	
	Computing the summands from (\ref{s1}) to (\ref{s4}), we obtain
	
	\begin{align}
		&(n-1) \Big(\frac{1}{(2N+1)^2} \sum_{|l|<N} l^2 \frac{4\pi^2}{T^2} +\frac{1}{(2N+1)^2} \sum_{l\neq l^{\prime}} ll^{\prime}  \frac{4\pi^2}{T^2} \mathrm{e}^{-\mathrm{i}\frac{2\pi}{n}(l-l^{\prime})} \Big) \nonumber\\
		&\times (9\E[\zeta^4]\E[\zeta^2]+\E[\zeta^6]+6\E[\zeta^2]^3-4\E[\zeta^3]^2) \label{s11}\\
		+& (n-1)\Big(\frac{1}{(2N+1)^2} \sum_{|l|<N} l^2 \frac{4\pi^2}{T^2} +\frac{1}{(2N+1)^2} \sum_{l\neq l^{\prime}} ll^{\prime}  \frac{4\pi^2}{T^2} \mathrm{e}^{\mathrm{i}\frac{2\pi}{n}(l-l^{\prime})} \Big) \nonumber\\ &\times (13\E[\zeta^4]\E[\zeta^2]+\E[\zeta^6]+2\E[\zeta^2]^3-4\E[\zeta^3]^2) \label{s22}\\
		+& (n-1) \Big(\frac{1}{(2N+1)^2} \sum_{|l|<N} l^2 \frac{4\pi^2}{T^2} \mathrm{e}^{-\mathrm{i}\frac{4\pi}{n}l} +\frac{1}{(2N+1)^2} \sum_{l\neq l^{\prime}} ll^{\prime}  \frac{4\pi^2}{T^2} \mathrm{e}^{-\mathrm{i}\frac{2\pi}{n}(l^{\prime}+l)} \Big) \nonumber\\
		&\times ( \E[\zeta^3]^2-\E[\zeta^6] -6\E[\zeta^4]\E[\zeta^2]-9\E[\zeta^2]^3) \label{s33}\\
		+& (n-1) \Big(\frac{1}{(2N+1)^2} \sum_{|l|<N} l^2 \frac{4\pi^2}{T^2} \mathrm{e}^{+\mathrm{i}\frac{4\pi}{n}l} +\frac{1}{(2N+1)^2} \sum_{l\neq l^{\prime}} ll^{\prime}  \frac{4\pi^2}{T^2} \mathrm{e}^{\mathrm{i}\frac{2\pi}{n}(l^{\prime}+l)} \Big) \nonumber\\
		& \times (\E[\zeta^3]^2-\E[\zeta^6] -6\E[\zeta^4]\E[\zeta^2]-9\E[\zeta^2]^3) \label{s44}
	\end{align}
	
	The terms where the indices $l$ and $l^{\prime}$ appear in (\ref{s11}) and (\ref{s33}) has to be summed up following the rule highlighted in Table \ref{TabL0}. The same applies for the terms where the indices $l$ and $l^{\prime}$ appear in (\ref{s22}) and (\ref{s44}).

	\begin{table}
		\footnotesize
		\centering.
		\begin{tabular}{|c|c|}
			\hline
			\textrm{Indices belonging to (\ref{s11})} & \textrm{Indices belonging to (\ref{s33})}\\
			\hline
			$\color{red}{l>0,l^{\prime}>0}$ &  $\color{red}{l>0, l^{\prime}<0}$ \\
			\hline
			$\color{red}{l<0, l^{\prime}<0}$ &  $\color{red}{l<0,l^{\prime}>0}$  \\
			\hline
			$\color{blue}{l>0,l^{\prime}<0}$ &  $\color{blue}{l>0,l^{\prime}>0}$ \\
			\hline
			$\color{blue}{l<0,l^{\prime}>0}$ & $\color{blue}{l<0,l^{\prime}<0}$   \\
			\hline
		\end{tabular}
		\caption{Summing rule: the terms has to be first summed according to the indices present in each row, then the resulting addends has to be summed with respect to the indices grouped by color.}
		\label{TabL0}
	\end{table}
	
	The summands from (\ref{s1}) to (\ref{s4}) are then equal to
	\begin{align}
		&(n-1) \frac{1}{(2N+1)^2} \sum_{|l|<N} l^2 \frac{4\pi^2}{T^2} (22\E[\zeta^4]\E[\zeta^2]+2\E[\zeta^6]+8\E[\zeta^2]^3-8\E[\zeta^3]^2) \label{s1f}\\
		&+ (n-1) \frac{1}{(2N+1)^2} \sum_{|l|<N} l^2 \frac{4\pi^2}{T^2} 2\cos(\frac{4\pi}{n}l) (\E[\zeta^3]^2-\E[\zeta^6] -6\E[\zeta^4]\E[\zeta^2]-9\E[\zeta^2]^3) \label{s2f}\\
		&+ (n-1) \frac{1}{(2N+1)^2} \sum_{l,l^{\prime}>0} l l^{\prime} \frac{4\pi^2}{T^2} \sin(\frac{2\pi}{n}l) \sin(\frac{2\pi}{n}l^{\prime}) (34\E[\zeta^4]\E[\zeta^2]+4\E[\zeta^6] \nonumber\\
		& +26 \E[\zeta^2]^3-5\E[\zeta^3]^2). \label{s3f}
	\end{align}
	If $N=n^{\frac{1}{\beta}}$ with $\beta > \frac{log(n)}{\log(n)-log(8)}$ then  $0\leq \cos(\frac{4\pi}{n}l) \leq 1$ for all $|l|\leq N$. If $\beta> \frac{log(n)}{log(n)-log(2)}$ then $0\leq \sin(\frac{2\pi}{n}l) \leq 1$. We have that the choice of $\beta> \frac{log(n)}{\log(n)-log(8)}$ is implied by the ratio in (\ref{hyp2}).
	In conclusion, the term (\ref{s3f}) is greater than or equal to zero and the sum between (\ref{s1f}) and (\ref{s2f}) is  greater than or equal to
	\begin{equation}
		\label{fin1}
		(n-1) \frac{1}{(2N+1)^2} \sum_{|l|<N} l^2 \frac{4\pi^2}{T^2} (22\E[\zeta^4]\E[\zeta^2]+2\E[\zeta^6]+8\E[\zeta^2]^3-8\E[\zeta^3]^2),
	\end{equation}
	if $(\E[\zeta^3]^2-\E[\zeta^6] -6\E[\zeta^4]\E[\zeta^2]-9\E[\zeta^2]^3) >0$ and
	greater than or equal to
	\begin{equation}
		\label{fin2}
		(n-1) \frac{1}{(2N+1)^2} \sum_{|l|<N} l^2 \frac{4\pi^2}{T^2} (10\E[\zeta^4]\E[\zeta^2]-10\E[\zeta^2]^3-6\E[\zeta^3]^2),
	\end{equation}
	if $(\E[\zeta^3]^2-\E[\zeta^6] -6\E[\zeta^4]\E[\zeta^2]-9\E[\zeta^2]^3) <0$.
	
	Let us now analyse the summands from (\ref{n1}) to (\ref{n3}). They are equal to
	
	\begin{align}
		&\sum_{i,j} \Big(\frac{1}{(2N+1)^2} \sum_{|l|<N} l^2 \frac{4\pi^2}{T^2} +\frac{1}{(2N+1)^2} \sum_{l\neq l^{\prime}} ll^{\prime}  \frac{4\pi^2}{T^2} \mathrm{e}^{-\mathrm{i}\frac{2\pi}{T}(l-l^{\prime})(t_i-t_j)} \Big) \nonumber\\
		& \times (2 \E[\delta_j^2] \E[\zeta^2] + 2\E[\delta_i^2]\E[\zeta^4] +6 \E[\delta_i^2]\E[\zeta^2]^2 + 8\E[\delta_i^2\delta_j^2] \E[\zeta^2] ) \label{n11} \\
		&+ \sum_{i,j:|i-j|\neq 1} \Big(\frac{1}{(2N+1)^2} \sum_{|l|<N} l^2 \frac{4\pi^2}{T^2} +\frac{1}{(2N+1)^2} \sum_{l\neq l^{\prime}} ll^{\prime}  \frac{4\pi^2}{T^2} \mathrm{e}^{-\mathrm{i}\frac{2\pi}{T}(l-l^{\prime})(t_i-t_j)} \Big) \nonumber\\
		&\times 16 \E[\delta_j^2] \E[\zeta^2]^2 \label{n21}\\
		&+\sum_{i,j:|i-j|=1} \Big(\frac{1}{(2N+1)^2} \sum_{|l|<N} l^2 \frac{4\pi^2}{T^2} +\frac{1}{(2N+1)^2} \sum_{l\neq l^{\prime}} ll^{\prime}  \frac{4\pi^2}{T^2} \mathrm{e}^{-\mathrm{i}\frac{2\pi}{T}(l-l^{\prime})(t_i-t_j)} \Big) \nonumber\\
		&\times \E[\delta_j^2](12 \E[\zeta^2]^2+ 4\E[\zeta^4]) \label{n22}\\
		&+ \sum_{i,j: |i-j|\neq 1} \Big(\frac{1}{(2N+1)^2} \sum_{|l|<N} l^2 \frac{4\pi^2}{T^2} +\frac{1}{(2N+1)^2} \sum_{l\neq l^{\prime}} ll^{\prime}  \frac{4\pi^2}{T^2} \mathrm{e}^{-\mathrm{i}\frac{2\pi}{T}(l-l^{\prime})(t_i-t_j)} \Big) \nonumber\\
		&\times (4\E[\zeta^4]\E[\zeta^2]+12\E[\zeta^2]^3). \label{n33}
	\end{align}

	\begin{table}
		\footnotesize
		\centering.
		\begin{tabular}{|c|c|}
			\hline
			
			& $l,l^{\prime}$   \\
			\hline
			$s$ &  $\color{blue}{ll^{\prime} (\cos(2\pi/n(ls-l^{\prime}s)) + i \sin(2\pi/n(ls-l^{\prime}s))} $   \\
			\hline
			$-s$ &  $\color{blue}{ll^{\prime} (\cos(2\pi/n(-ls+l^{\prime}s)) + i \sin(2\pi/n(-ls+l^{\prime}s))}$   \\
			\hline
			& $ l,-l^{\prime}$   \\
			\hline
			$s$ & $\color{red}{-ll^{\prime} (\cos(2\pi/n(ls+l^{\prime}s)) + i \sin(2\pi/n(ls+l^{\prime}s))}$   \\
			\hline
			$-s$ &   $\color{red}{-ll^{\prime} (\cos(2\pi/n(-ls-l^{\prime}s)) + i \sin(2\pi/n(-ls-l^{\prime}s))}$ \\
			\hline
			& $-l,-l^{\prime}$ \\
			\hline
			$s$ & $\color{blue}{ll^{\prime} (\cos(2\pi/n(-ls+l^{\prime}s)) + i \sin(2\pi/n(-ls+l^{\prime}s))}$\\
			\hline
			$-s$ & $\color{blue}{ll^{\prime} (\cos(2\pi/n(ls-l^{\prime}s)) + i \sin(2\pi/n(ls-l^{\prime}s))}$\\
			\hline
			& $-l,l^{\prime}$ \\
			\hline 
			$s$ & $\color{red}{-ll^{\prime} (\cos(2\pi/n(-ls-l^{\prime}s)) + i \sin(2\pi/n(-ls-l^{\prime}s))}$ \\
			\hline
			$-s$ & $\color{red}{-ll^{\prime} (\cos(2\pi/n(ls+l^{\prime}s)) + i \sin(2\pi/n(ls+l^{\prime}s))}$\\
			\hline
		\end{tabular}
		\caption{Coefficients appearing in the summands from (\ref{n11}) to (\ref{n33}) with respect to the indices $l,l^{\prime}$, $s$. Note that $t_i-t_j=\frac{2\pi}{n}s$ and all the indices are considered positive constants.}
		\label{TabL3}
	\end{table}
	
	\begin{table}
		\footnotesize
		\centering.
		\begin{tabular}{|c|c|}
			\hline
			
			& $l,l^{\prime}$   \\
			\hline
			$s$ &  $\color{blue}{ll^{\prime} (2\cos(2\pi/n(ls-l^{\prime}s)) )} $  \\
			\hline
			$-s$ &  $\color{blue}{ll^{\prime} (2\cos(2\pi/n(-ls+l^{\prime}s)) )}$ \\
			\hline
			& $ l,-l^{\prime}$   \\
			\hline
			$s$ & $\color{red}{-ll^{\prime} (2\cos(2\pi/n(ls+l^{\prime}s)))}$   \\
			\hline
			$-s$ &   $\color{red}{-ll^{\prime} (2\cos(2\pi/n(-ls-l^{\prime}s)) )}$  \\
			\hline
		\end{tabular}
		\caption{Coefficients appearing in the summands from (\ref{n11}) to (\ref{n33}) with respect to the indices $l,l^{\prime}$, $s$. Note that $t_i-t_j=\frac{2\pi}{n}s$ and all the indices are considered positive constants.}
		\label{TabL4}
	\end{table}

	We first take care of the sum with respect to the indices $i,j,l$ and $l^{\prime}$ appearing in the terms (\ref{n11}) to (\ref{n33}). To simply explain the calculations below, let us consider from now on that the indices $l$ and $l^{\prime}$ are positive and that there exists an $s=1,\ldots,n-1$ such that if $t_i> t_j$, $t_i-t_j=s \frac{2\pi}{n}$. We do not consider $s=0$ in the calculations below because $D^{\prime}_N(t_i-t_i)=0$.
	In Table \ref{TabL3}, we find, for fixed values of $s,l,l^{\prime}$, all the possible combination of the indices and the expression of the terms $ll^{\prime} \mathrm{e}^{-\mathrm{i}\frac{2\pi}{T}(l-l^{\prime})(t_i-t_j)}$ appearing in the summands. The blue and the red elements appear in (\ref{n11}), (\ref{n21}), (\ref{n22}) and(\ref{n33}), respectively, the same number of times. We first sum each row of Table \ref{TabL3}. We then obtain, Table \ref{TabL4}. Summing up the blue and red terms obtained for $s$ and $-s$, respectively, we have

	\begin{align}
		&\sum_{i,j} \frac{1}{(2N+1)^2} \sum_{|l|<N} l^2 \frac{4\pi^2}{T^2} \, (2 \E[\delta_j^2] \E[\zeta^2] + 2\E[\delta_i^2]\E[\zeta^4] +6 \E[\delta_i^2]\E[\zeta^2]^2 \nonumber\\
		& + 8\E[\delta_i^2\delta_j^2] \E[\zeta^2] )  \label{n11f0}\\
		& \sum_{i=1}^{n-1} \sum_{j=0}^{i-1} \frac{1}{(2N+1)^2} \sum_{l,l^{\prime}>0} ll^{\prime}  \frac{4\pi^2}{T^2} \, 4 \sin(\frac{2\pi}{n}ls) \sin(\frac{2\pi}{n}l^{\prime}s) \nonumber\\
		&\times (2 \E[\delta_j^2] \E[\zeta^2] + 2\E[\delta_i^2]\E[\zeta^4] +6 \E[\delta_i^2]\E[\zeta^2]^2 + 8\E[\delta_i^2\delta_j^2] \E[\zeta^2] ) \label{n11f1} \\
		+& \sum_{i,j:|i-j|\neq 1} \frac{1}{(2N+1)^2} \sum_{|l|<N} l^2 \frac{4\pi^2}{T^2} 16 \E[\delta_j^2] \E[\zeta^2]^2 \label{n22f0}\\
		+& \sum_{i=2}^{n-1} \sum_{j=0}^{i-2} \frac{1}{(2N+1)^2} \sum_{l,l^{\prime}>0} ll^{\prime}  \frac{4\pi^2}{T^2} \, 4 \sin(\frac{2\pi}{n}ls) \sin(\frac{2\pi}{n}l^{\prime}s)  16 \E[\delta_j^2] \E[\zeta^2]^2 \label{n22f1}\\
		&+\sum_{i,j:|i-j|=1} \frac{1}{(2N+1)^2} \sum_{|l|<N} l^2 \frac{4\pi^2}{T^2} \, \E[\delta_j^2](12 \E[\zeta^2]^2+ 4\E[\zeta^4]) \label{n44f0}\\
		&+  (n-1)  \frac{1}{(2N+1)^2} \sum_{l,l^{\prime}>0} ll^{\prime}  \frac{4\pi^2}{T^2} \, 4 \sin(\frac{2\pi}{n}l) \sin(\frac{2\pi}{n}l^{\prime})  \, \E[\delta_j^2](12 \E[\zeta^2]^2+ 4\E[\zeta^4]) \label{n44f1}\\
		+& \sum_{i,j: |i-j|\neq 1} \frac{1}{(2N+1)^2} \sum_{|l|<N} l^2 \frac{4\pi^2}{T^2} \, (4\E[\zeta^4]\E[\zeta^2]+12\E[\zeta^2]^3) \label{n33f0}\\
		+&  \sum_{i=2}^{n-1} \sum_{j=0}^{i-2} \frac{1}{(2N+1)^2} \sum_{l,l^{\prime}>0} ll^{\prime}  \frac{4\pi^2}{T^2} \, 4 \sin(\frac{2\pi}{n}ls) \sin(\frac{2\pi}{n}l^{\prime}s) \, (4\E[\zeta^4]\E[\zeta^2]+12\E[\zeta^2]^3) \label{n33f1}
	\end{align}
	
	If $N=n^{\frac{1}{\beta}}$ such that $\beta > \frac{log(n)}{log(n)-log(2(n-1))}$ then
	$0 \leq \sin(\frac{2\pi}{n} l s) \leq 1$ for $s=1,\ldots,n-1$ and $l>0$. The latter is straightforwardly implied by (\ref{hyp2}), being $\frac{log(n)}{log(n)-log(2(n-1))}$ negative. Therefore, the summands (\ref{n11f1}), (\ref{n22f1}), (\ref{n44f1}) and (\ref{n33f1}) are greater than or equal to zero.
	
	In conclusion, (\ref{mse_dom}) is possibly greater than or equal to two sums.
	The first one is
	\begin{align}
		&\sum_{i,j} \frac{1}{(2N+1)^2} \sum_{|l|<N} l^2 \frac{4\pi^2}{T^2} \, (2 \E[\delta_j^2] \E[\zeta^2] + 2\E[\delta_i^2]\E[\zeta^4] +6 \E[\delta_i^2]\E[\zeta^2]^2   \nonumber\\
		& + 8\E[\delta_i^2\delta_j^2] \E[\zeta^2] ) + \sum_{i,j:|i-j|\neq 1} \frac{1}{(2N+1)^2} \sum_{|l|<N} l^2 \frac{4\pi^2}{T^2} 16 \E[\delta_j^2] \E[\zeta^2]^2 \nonumber\\
		&+\sum_{i,j:|i-j|=1} \frac{1}{(2N+1)^2} \sum_{|l|<N} l^2 \frac{4\pi^2}{T^2} \, \E[\delta_j^2](12 \E[\zeta^2]^2+ 4\E[\zeta^4]) \label{nowhy} \\
		+& \sum_{i,j: |i-j|\neq 1} \frac{1}{(2N+1)^2} \sum_{|l|<N} l^2 \frac{4\pi^2}{T^2} \, (4\E[\zeta^4]\E[\zeta^2]+12\E[\zeta^2]^3) \nonumber\\
		+& (n-1) \frac{1}{(2N+1)^2} \sum_{|l|<N} l^2 \frac{4\pi^2}{T^2} (22\E[\zeta^4]\E[\zeta^2]+2\E[\zeta^6]+8\E[\zeta^2]^3-8\E[\zeta^3]^2) \nonumber ,
	\end{align}
	and the second one is
	\begin{align}
		&\sum_{i,j} \frac{1}{(2N+1)^2} \sum_{|l|<N} l^2 \frac{4\pi^2}{T^2} \, (2 \E[\delta_j^2] \E[\zeta^2] + 2\E[\delta_i^2]\E[\zeta^4] +6 \E[\delta_i^2]\E[\zeta^2]^2 + 8\E[\delta_i^2\delta_j^2] \E[\zeta^2] )  \nonumber\\
		+& \sum_{i,j:|i-j|\neq 1} \frac{1}{(2N+1)^2} \sum_{|l|<N} l^2 \frac{4\pi^2}{T^2} 16 \E[\delta_j^2] \E[\zeta^2]^2 \nonumber
	\end{align}
	\begin{align}
		&+\sum_{i,j:|i-j|=1} \frac{1}{(2N+1)^2} \sum_{|l|<N} l^2 \frac{4\pi^2}{T^2} \, \E[\delta_j^2](12 \E[\zeta^2]^2+ 4\E[\zeta^4]) \label{nowhy2} \\
		+& \sum_{i,j: |i-j|\neq 1} \frac{1}{(2N+1)^2} \sum_{|l|<N} l^2 \frac{4\pi^2}{T^2} \, (4\E[\zeta^4]\E[\zeta^2]+12\E[\zeta^2]^3) \nonumber\\
		+& (n-1) \frac{1}{(2N+1)^2} \sum_{|l|<N} l^2 \frac{4\pi^2}{T^2} (10\E[\zeta^4]\E[\zeta^2]-10\E[\zeta^2]^3-6\E[\zeta^3]^2)
		\nonumber.
	\end{align}

	Note that because of Assumption (H2), the terms $\E[\delta_i^2]$ and $\E[\delta_i^2\delta_j^2]$ are positive and finite constants. The behaviour of the sums (\ref{nowhy}) and (\ref{nowhy2}) is ruled by their first summands. In fact,
	\[
	\sum_{i,j} \frac{1}{(2N+1)^2} \sum_{|l|\leq N} l^2 \frac{4\pi^2}{T^2}= \frac{n^2}{(2N+1)^2} \frac{4 \pi^2}{T^2} \Big( \frac{N^3}{3}+\frac{N^2}{2}+\frac{N}{6} \Big),
	\]
	which diverges as $n, N \to \infty$.
\end{proof}

\end{document}